\documentclass[11pt, oneside,notitlepage]{article}
\usepackage{STpaperformat}
\usepackage{STpapersymbol}
\usepackage{color}

\usepackage[utf8]{inputenc}
\usepackage[T1]{fontenc}

\usepackage{amsmath,amssymb} 
\usepackage{bm}  
\usepackage{mathbbol}

\usepackage{fourier}
\usepackage[scaled=0.875]{helvet}

\usepackage[margin=0.8in]{geometry}                		
\usepackage{graphicx}				

\usepackage{flafter}  
\usepackage{type1cm}

\usepackage{float}

\usepackage[
bookmarks=true,
bookmarksnumbered=true,
bookmarkstype=toc,
colorlinks,
breaklinks,
pdfencoding=auto,
psdextra,
pdfauthor=Seiichiro Tani,%
]{hyperref}
\definecolor{dullmagenta}{rgb}{0.4,0,0.4}   
\definecolor{darkblue}{rgb}{0,0,0.4}
\hypersetup{linkcolor=red,citecolor=blue,filecolor=dullmagenta,urlcolor=darkblue} 

\def\equationautorefname~#1\null{\textrm{Eq.~(#1)}\null}
\def\figureautorefname~#1\null{\textrm{Fig.~#1}\null}
\def\tableautorefname~#1\null{\textrm{Tab.~#1}\null}
\def\sectionautorefname~#1\null{\textrm{Sec.~#1\;}\null}
\def\subsectionautorefname~#1\null{\textrm{Sec.~#1\;}\null}
\def\subsubsectionautorefname~#1\null{\textrm{Sec.~#1\;}\null}
\def\pageautorefname~#1\null{\textrm{page~#1\;} \null}

\author{Seiichiro Tani\footnote{seiichiro.tani@acm.org}\\
Department of Mathematics, School of Education\\
Waseda University
}
\date{}

\newcommand{\domain}{\mathrm{dom}}

\newcommand{\claw}{\mathsf{claw}}

\newcommand{\adeg}{\widetilde{\deg}}
\newcommand{\pSearch}{\mathsf{pSearch}}

\title{Approximate Degrees of Multisymmetric Properties\\
with Application to Quantum Claw Detection}

\begin{document}

\sloppy
\maketitle
\thispagestyle{empty}
\begin{abstract}
    The claw problem is central in the fields of theoretical computer science as well as cryptography.
    The goal of the problem is to detect the existence of a pair of inputs $x$ and $y$ (or find such a pair if they exist) such that $f(x)=g(y)$ for a given pair of funcitons $f$ and $g$ as oracles (i.e., black boxes).
    The optimal quantum query complexity of the problem is known to be $\Omega\left(\sqrt{G}+(FG)^{1/3} \right)$ for input functions $f\colon [F]\to Z$ and $g\colon [G]\to Z$. 
    However, the lower bound was proved when the range $Z$ is sufficiently large (i.e., $\abs{Z}=\Omega(FG)$). The current paper proves the lower bound holds even for every smaller range $Z$ with $\abs{Z}\ge F+G$. This implies that $\Omega\left(\sqrt{G}+(FG)^{1/3} \right)$ is tight for every such range.
    In addition, the lower bound
    $\Omega\left(\sqrt{G}+F^{1/3}G^{1/6}M^{1/6}\right)$ is provided for even smaller range $Z=[M]$ with every $M\in [2,F+G]$ by reducing the claw problem for $\abs{Z}= F+G$.
    The proof technique is general enough to apply to any $k$-symmetric property (e.g., the $k$-claw problem), i.e., the Boolean function $\Phi$ on the set of $k$ functions with different-size domains and
    a common range such that $\Phi$ is invariant under the permutations over each domain and the permutations over the range. More concretely,
    it generalizes Ambainis's argument [Theory of Computing, 1(1):37-46] 
    to the multiple-function case by using the notion of multisymmetric polynomials.
\end{abstract}
\clearpage
\setcounter{page}{1}
\pagestyle{plain}
\section{Introduction}
\subsection{Background}
Due to their simplicity and broad applicability, the collision detection/finding problem and its variant, the \emph{claw detection/finding problem}, are central problems in the fields of theoretical computer science 
(e.g.,\cite{BuhDurHeiHoyMagSanWol01CCC,IwaKaw03NG,AarShi04JACM,Amb07SICOMP,BelChiJefKotMag13ICALP,Zha15QIC})
as well as cryptography (e.g.,\cite{BraHOyTap98LATIN,JouLuc09ASIACRYPT,NikSas16ASIACRYPT,LiuZha19EUROCRYPTO,HosSasTanXag20TCS}).
For a given function $f$, the goal of the collision detection is to detect the existence of a pair of elements $(x,y)$ such that $f(x)=f(y)$
(Since the collision detection without any promise on the input $f$ is often called \emph{element distinctness}, we may use the latter term when appropriate.)
As a natural variant of the collision detection, we can define the \emph{claw detection} for two functions $f,g$ as the problem of detecting the existence of a pair $(x,y)$ such that $f(x)=g(y)$.  
When we are required to \emph{find} a collision (a claw) if it exists,
we refer to the corresponding ``search'' problem as the \emph{collision finding} (the \emph{claw finding}, respectively) in this paper. 

These problems are classically hard (i.e., time-consuming) to solve for large domains. This hardness forms the basis of the security of various cryptographic applications.
This fact in turn makes these problems appealing as targets with which to investigate quantum computational power: If one finds quantum algorithms that solve the problems very efficiently, this exhibits solid quantum advantages concerning practically essential problems; if one proves that the problem is very hard even against quantum computation, this gives theoretical evidence that the cryptographic applications are secure against quantum computers.

Interestingly, the actual situation is somewhere in between. For these problems, quantum speed-ups are possible but at most polynomial. This has motivated the invention of various essential techniques for improving quantum upper bounds (i.e., quantum algorithms) and lower bounds.

For element distinctness,  
Buhrman et al.~\cite{BuhDurHeiHoyMagSanWol01CCC} proved 
the upper bound $O(n^{3/4})$ and the lower bound $\Omega(\sqrt{n})$
on the quantum query complexity for $f\colon [n]\to Z$ for a finite set $Z$,
where $[n]$ denotes the set $\set{1,\dots, n}$.
Aaronson and Shi~\cite{AarShi04JACM} improved the lower bound to
$\Omega(n^{2/3})$
by introducing a novel technique to the polynomial method to bound the degrees of multi-variate polynomials that approximate symmetric Boolean functions. Ambainis~\cite{Amb07SICOMP} then provided the breakthrough upper bound $O(n^{2/3})$ by developing a general framework of quantum search based on quantum walk. His work was followed by many generalizations and variants that form one of the mainstreams of quantum algorithms (e.g.,\cite{Sze04FOCS,MagNayRolSan11SICOMP,AmbGilJefKok20STOC,JefZur23STOC}). 

As for claw detection for $f\colon [F]\to Z$ and $g\colon [G]\to Z$, where
$F$ and $G$ are positive integers such that
$F\le G$ and $Z$ is a finite set, 
Buhrman et al.~\cite{BuhDurHeiHoyMagSanWol01CCC} proved 
that the quantum query complexity is $\Omega(\sqrt{G})$ and $O(\sqrt{G}+F^{1/2}G^{1/4})$.
Zhang~\cite{Zha05COCOON} then proved the lower bound $\Omega(\sqrt{G}+(FG)^{1/3})$ on the quantum query complexity by 
providing the reduction from collision detection with the promise that the input function with domain $[FG]$ is either one-to-one or two-to-one and using the lower bound $\Omega((FG)^{1/3})$~\cite{AarShi04JACM} for that promised case. Tani~\cite{Tan09TCS} then proved the matching upper bound $O(\sqrt{G}+(FG)^{1/3})$ by running a quantum-walk-based search on a particular graph product of two Johnson graphs corresponding to the two input functions.

These tight bounds depend only on the domain size(s) of the input function(s). This means that these bounds are tight if the function range is sufficiently large. Indeed, the lower bound $\Omega (n^{2/3})$ of element distinctness was initially proved when the range size is $\Omega(n^2)$~\cite{AarShi04JACM}. 
In some cases of smaller range size, however, it is still non-trivial to detect a collision (e.g., a collision for a given function $f\colon [n]\to [n]$ since $f$ can be one-to-one). Moreover, in some applications, such as cryptographic ones, the range size is often known and can be a relatively small value. Ambainis~\cite{Amb05ToC} and Kutin~\cite{Kut05ToC} independently proved the same lower bound $\Omega (n^{2/3})$ holds even for the range of size $n$.  In particular, Ambainis’s argument is so general that it can be applied to any function $\Phi$ from a set of functions $f$ to $\set{0,1}$ if $\Phi$ is invariant under the permutations $\pi$ and $\sigma$ over the domain and range, respectively, of $f$, i.e., $\Phi(f)=\Phi(\sigma f\pi)$ (e.g., if $\Phi$ tells the existence of a collision for $f$, then it is easy to see that $\Phi(f)=\Phi(\sigma f\pi)$).

As for claw detection, Zhang’s lower bound holds if the range size of input functions is $\Omega(FG)$. This condition seems very strong, considering that it is still non-trivial  to detect a claw for a given pair of functions 
$f\colon [F]\to [F+G]$ and $g\colon [G]\to [F+G]$ with $F\le G$,
since the images of $f$ and $g$ can be disjoint.  Therefore, it is natural to ask whether the lower bound 
$\Omega(\sqrt{G}+(FG)^{1/3})$ still
holds for a smaller range, such as $[F+G]$, and, more ambitiously, whether a multi-function version of Ambainis’s argument exists. From a practical point of view, relaxing the condition on the function range for tight bounds could contribute to more accurate security analysis of post-quantum ciphers, such as those based on supersingular elliptic curve isogenies~\cite{JaqSch19CRYPTO}, one form (SIKE) of which had been considered until the final round of the post-quantum standardization at NIST.

\subsection{Polynomial Method}\label{subsec:polynomialmethod}
There are two primary lower-bound techniques for quantum query complexity: the polynomial method \cite{BeaBuhCleMosWol01JACM} and the adversary method \cite{Amb02JCSS,HoyLeeSpa07STOC}. 
Although the polynomial method does not necessarily give optimal lower bounds~\cite{Amb06JCSS}, it often gives non-trivial lower bounds that were (or are) not previously known to be obtainable using other lower-bound techniques (e.g., Refs.~\cite{AarShi04JACM,BunKotTha20TOC}). 
The highlight of the method is the claim that the minimum degree of a polynomial that approximates a Boolean function $\Phi$ is at most twice the number of quantum queries required to compute $\Phi$. This allows us to use the results of polynomial approximation theory to obtain lower bounds on the quantum query complexity of $\Phi$.
More concretely, let $f$ be any function from $[n]$ to $[m]$. Then, $f$ can be identified 
with $nm$ Boolean variables $\{x_{ij}:i\in\left[n\right],j\in\left[m\right]\}$ such that $x_{ij}=1$ 
if and only if $f\left(i\right)=j$. Let $\Phi$ be a Boolean function that tells whether input $f$ has a certain property (e.g., the property that $f$ has at least one collision). 
Then, $\Phi$ is a function in the variables $\{x_{ij}\}$.  
The polynomial method claims that the minimum degree of a polynomial that approximates $\Phi$ within error $\epsilon$ is at most twice the number of quantum queries required 
to compute $\Phi$ with error probability at most $\epsilon$. 
For $m\geq n$, suppose that $\Phi^\prime$ 
is the restriction of $\Phi$ to the case 
where the range of input function is $\left[n\right]$.
Then, it is easy to see that if $d$ is the minimum degree of a polynomial that approximates $\Phi$, then the minimum degree of a polynomial that approximates $\Phi^\prime$ is at most $d$ (since $\Phi^\prime$ is obtained by setting $x_{ij}=0$ for all $i\in [n]$ and $j\in\left[m\right]\setminus\left[n\right]$ in $\Phi$). 
Ambainis~\cite{Amb05ToC} proved that the converse is also true if $\Phi$ is a symmetric property, i.e., a Boolean function invariant under the permutations over domain $[n]$ and range $\left[m\right]$, respectively, of $f$.
For instance, following the lower-bound proof of 
element distinctness by Aaronson and Shi~\cite{AarShi04JACM},
its minimum approximate degree is $\Omega\left(n^{2/3}\right)$
if the range size of input functions is $m\in \Omega(n^2)$. 
Thus, the polynomial method together with Ambainis's argument implies that the quantum query complexity of the element distinctness problem is $\Omega\left(n^{2/3}\right)$ even if the function range is $\left[n\right]$. 
This relaxes the condition (i.e., $m=\Omega(n^2)$) on the range size of input functions
required in Ref.~\cite{AarShi04JACM} required to obtain the optimal lower bound.

\subsection{Our Contribution}\label{sec:outcontribution}
Our main contribution is to prove a multi-function version of Ambainis’s argument. 
As an application of our argument to two functions 
$f\colon [F]\to [M]$ and $g\colon [G]\to [M]$, we show that the optimal quantum lower bound for claw detection proved for range size $FG$ still holds for smaller range size $F+G$. 
Moreover, using this new lower bound, we obtain a lower bound for an even smaller range $[M]$ with $2\le\left|M\right|<F+G$. In the following, we summarize our results more concretely. 
To simplify the notation, 
we represent by $\calR^{\calD}$ the set of functions from $\calD$ to $\calR$ for any ordered sets $\calD$ and $\calR$. For instance, 
if $f$ is a function from $[F]$ to $[M]$, then $f\in [M]^{[F]}$.

\subsubsection{Approximate Degree of Multisymmetric Properties}
For simplicity, we first consider the case where two functions $f\colon\left[F\right]\rightarrow\left[M\right]$ 
and $g\colon\left[G\right]\rightarrow\left[M\right]$ are given as input. 
Then, $f$ can be identified with $FM$ Boolean variables $\{x_{ij}:i\in\left[F\right],j\in\left[M\right]\}$ 
such that $x_{ij}=1$ if and only if $f\left(i\right)=j$. 
Similarly, $g$ can be identified with $GM$ Boolean variables 
$\{y_{kj}:k\in\left[F\right],j\in\left[M\right]\}$. 
For $\calD\subseteq [M]^{[F]} \times [M]^{[G]}$,
Let $\Phi\colon \calD\to \set{0,1}$ be a Boolean function that tells 
whether input pair $(f,g)$ in $\calD$ has a certain property 
(e.g., the claw property that $f$ and $g$ have at least one claw). 
Then, $\Phi$ is a function in the $FM+GM$ variables $\{x_{ij}\}\cup \{y_{kj}\}$.
We say that $\Phi$ is bisymmetric, if $\Phi$ 
is invariant under the permutations over domains $\left[F\right]$ and $[G]$, respectively, and the common range $\left[M\right]$ of $f$ and $g$, i.e., 
$\Phi\left(f,g\right)=\Phi\left(\sigma f\pi_F,\sigma g\pi_G\right)$,
where $\pi_F$ and $\pi_G$ are permutations 
over $[F]$ and $[G]$, respectively, and $\sigma$ is a permutation over $\left[M\right]$.
For instance, if $\Phi$ is the claw property, then $\Phi$ is bi-symmetric.

\begin{theorem}\label{th:bisymmetricproperty_intro}
    For $F,G,M\in \Natural$ such that $M\ge F+G$, 
    let 
       $\Phi\colon [M]^{[F]}\times [M]^{[G]}\to \set{0,1}$ 
    be a bisymmetric Boolean function of a pair of functions 
    $f\colon [F]\to [M]$ and $g\colon [G]\to [M]$. 
    In addition, for every $M'$ with $F+G\le M'<M$,
    $\Phi'\colon [M']^{[F]}\times [M']^{[G]}\to \set{0,1}$ be
    the restriction of $\Phi$ to
    the case where the input functions are $f\colon [F]\to [M']$ 
    and $g\colon [G]\to [M']$.
        Then, for $0<\epsilon <1/2$, the minimum degree of a polynomial that $\epsilon$-approximates $\Phi$ is equal to the minimum degree of a polynomial that $\epsilon$-approximates $\Phi'$.
\end{theorem}
With the promise that the union of the images of input functions $f$ and $g$, i.e., ${f([F])\cup g([G])}$, is upper-bounded in size 
by $\mathfrak{U}$, \autoref{th:bisymmetricproperty_intro} also holds for every $M'$ with $\mathfrak{U}\le M' <M$ with a slight modification to the proof.
For instance, if both $f$ and $g$ are two-to-one,
then $\abs{f([F])\cup g([G])}\le (F+G)/2$, and thus
\autoref{th:bisymmetricproperty_intro} holds for every 
$M'$ with $(F+G)/2\le M' <M$.
Since $\abs{f([F])\cup g([G])}\le F+G$ always holds, we obtain \autoref{th:bisymmetricproperty_intro} in the above form (i.e., without any promise on the image sizes).

The definition of bisymmetric properties can be generalized in a straightforward way:
Let $\Phi\colon \calF_1 \times\cdots\times \calF_k \to \set{0,1}$
be a Boolean function over the set of $k$-tuples of functions in $\calF_1\times\cdots \times \calF_k$, where $\calF_{\ell}\colon [F_{\ell}]\to [M]$
for $F_{\ell},M\in \Integer_+$.
As in the case of two functions,
we can see that
$\Phi$ depends on Boolean variables $x^{(\ell)}_{ij}\in\{0,1\}$ for 
        $(\ell,i,j)\in [k]\times [F_{\ell}]\times [M]$ such that $x^{(\ell)}_{ij}=1$ if and only if $f_{\ell}\left(i\right)=j$.
We say that $\Phi$ is $k$-symmetric, if $\Phi(f_1,\dots,f_k)$ is invariant under the permutations over $[F_\ell]$ for each $\ell\in [k]$, respectively, and the common range $[M]$ of $f_{\ell}$'s, i.e., 
\begin{equation*}
\Phi(f_1,\dots,f_k)=\Phi\left(\sigma f_1\pi_{F_1},\dots,\sigma f_k\pi_{F_k}\right),  
\end{equation*}
for every permutation $\pi_{F_i}$ over $[F_i]$ for $i\in [k]$ and every permutation $\sigma$ over $[M]$. 
Then, we have a generalization of \autoref{th:bisymmetricproperty_intro}.

\begin{theorem}\label{th:multisymmetricproperty_intro}
    Let $k\ge 2$.
    For $F_\ell\ (\ell\in [k])$ and $M\in \Natural$ such that $M\ge F_1+\dots +F_k$, 
    $\Phi:\calF_1\times\dots \times \calF_k\to \set{0,1}$
    be a $k$-symmetric Boolean function over the set of $k$-tuples of functions $(f_1,\dots,f_k)\in \calF_1\times \dots \times \calF_k$, 
    where $\calF_{\ell}$ denotes $[M]^{[F_\ell]}$ for each $\ell$.
    In addition, 
    for every ${M'}$ such that $F_1+\dots+F_k\le {M'}< M$,
    let 
    $\Phi':\calF'_1\times\dots \times \calF'_k\to \set{0,1}$
    be 
    the restriction of $\Phi$ to the domain $[{M'}]^{[F_1]}\times\dots \times [{M'}]^{[F_k]}$, i.e., 
    the $k$-symmetric Boolean function over the set of $k$-tuples of functions 
    $(f'_1,\dots,f'_k)\in \calF'_1\times\dots \times \calF'_k$,
    where $\calF'_\ell$ denotes $[{M'}]^{[F_\ell]}$ for each $\ell\in [k]$,
    such that $\Phi'(f'_1,\dots,f'_k)=\Phi(f_1,\dots,f_k)$ for all functions $f'_{1},\dots,f'_{k},f_1,\dots, f_{k}$
    satisfying $f'_{\ell}(x)=f_{\ell}(x)$ for all $x\in [F_\ell]$ and all $\ell\in [k]$.

    Then, the minimum degree of a polynomial that $\epsilon$-approximates $\Phi$ is equal to the minimum degree of a polynomial that $\epsilon$-approximates $\Phi'$
\end{theorem}
With the promise that the union of the images of input functions, i.e., $\bigcup_{\ell\in [k]}{f_{\ell}([F_\ell])}$, is upper-bounded in size 
by $\mathfrak{U}$, \autoref{th:multisymmetricproperty_intro} also holds for every $M'$ with $\mathfrak{U}\le M' <M$ with a slight modification to the proof. 
Slightly general versions of the above theorems are 
Theorems~\ref{th:bisymmetricproperty} and \ref{th:multisymmetricproperty}
proved
in \autoref{sec:DegreeOfMultisymmetricProperties}.

\subsubsection{Application to Claw Detection}

\begin{table}[t]
    \centering
    \begin{tabular}{c||c|c|c|c}\hline
        & $2\le M < \max\{2,(G/F)^2\}$ & $\max\{2,(G/F)^2\}\le M < F+G$ & $F+G\le M < FG$ &   $FG\le M$\\\hline\hline
        $F\le G\le F^2$ & \multicolumn{2}{c|}{$O\left(
        \min\left\{(FG)^{1/3},G^{1/2+\epsilon}M^{1/4}\right\}
        \right)$~\cite{Tan09TCS,AmbBalIra21TQC}} & \multicolumn{2}{c}{$O((FG)^{1/3})$~\cite{Tan09TCS}}\\\cline{2-5}
         & $\Omega(\sqrt{G})$~\cite{BuhDurHeiHoyMagSanWol01CCC}& $\Omega(F^{1/3}G^{1/6}M^{1/6})$ \textcolor{red}{[Th.\ref{th:ComplexityClawWithSmallerRange_intro}]} &$\Omega ((FG)^{1/3})$ \textcolor{red}{[Th.\ref{th:QueryCompClawWithSmallRange_intro}]} & $\Omega ((FG)^{1/3})$~\cite{Zha05COCOON}\\\hline
        $F^2<G$ & \multicolumn{4}{c}{$\Theta\left(\sqrt{G}\right)$~\cite{BuhDurHeiHoyMagSanWol01CCC}}\\\hline
    \end{tabular}
    \caption{Summary of quantum query complexity for claw detection for functions $f\colon [F]\to [M]$ and $g\colon [G]\to [M]$, where $F\le G$. Note that $(G/F)^2<F+G$ holds if $G\le F^2$ (otherwise, $\Theta(\sqrt{G})$ is tight for all $M\ge 2$).}
    \label{tab:claw}
\end{table}

We use \autoref{th:bisymmetricproperty_intro}
to derive a lower bound for claw detection with a small range.
For this, we need a lower bound on the minimum degree
of a polynomial that approximates the claw property 
(for input functions with a possibly large common range).
It turns out that 
for any constant $0<\epsilon<1/2$, the minimum degree of a polynomial that $\epsilon$-approximates the claw property for function pairs 
$f\colon [F]\to [FG]$ and $g\colon [G]\to [FG]$
is $\Omega\left((FG)^{1/3}\right)$.
By combining this with \autoref{th:bisymmetricproperty_intro},
we can obtain a lower bound on the minimum approximate degree of the claw property for the range size at least $F+G$
of input functions. 
\begin{theorem}\label{th:DegreeOfClawForSmallRange_intro}
Let $F,G\in \Natural$ be such that $F\le G$.
For any constant $0<\epsilon<1/2$ and any $M\ge F+G$, the minimum degree of a polynomial that $\epsilon$-approximates the claw property for a function pair in $[M]^{[F]}\times [M]^{[G]}$ is $\Theta(\sqrt{G}+(FG)^{1/3})$,
where the lower bound $\Omega(\sqrt{G})$ holds for all $M\ge 2$.
\end{theorem}
By the polynomial method, \autoref{th:DegreeOfClawForSmallRange_intro} implies
the optimal lower bound on the bounded-error quantum query complexity of both claw detection and claw finding, where the optimality follows from the quantum algorithm~\cite{Tan09TCS}
that \emph{finds} a claw.

\begin{theorem}\label{th:QueryCompClawWithSmallRange_intro}
Let $F,G\in \Natural$ be such that $F\le G$.
For every $M\ge F+G$,
the (query-)optimal quantum algorithm that computes a claw (or detects the existence of a claw) with a constant error probability
requires $\Theta(\sqrt{G}+(FG)^{1/3})$ queries for a given pair of functions in $[M]^{[F]}\times [M]^{[G]}$.
More precisely, the optimal quantum query complexity is $\Theta ((FG)^{1/3})$ 
for every $M\ge F+G$ if $F\le G\le F^2$ 
and $\Theta(\sqrt{G})$ for every $M\ge 2$ if $G>F^2$.
\end{theorem}

Ambainis, Balodis and Iraids~\cite{AmbBalIra21TQC} investigated upper and lower bounds for
the claw problem for input function pairs with much smaller ranges.
They proved the lower bound $\Omega\left(N^{1/2}M^{1/6}\right)$
when the two input functions has the same domain size $N$ and the range size $M$ with $2\le M< N$. Their proof is based on the lower bound
$\Omega(N^{2/3})$ on the number of quantum queries required to detect a claw for a given pair of functions in $[N]^{[N]}\times [N]^{[N]}$, which follows by reducing from element distinctness. However, if the domains of two input functions are different in size, their proof breaks down. 
By generalizing their proof using \autoref{th:QueryCompClawWithSmallRange_intro}, we obtain a lower bound for functions with different-size domains and very small common ranges.
\begin{theorem}\label{th:ComplexityClawWithSmallerRange_intro}
Let $F,G,M\in \Natural$ be such that $F\le G\le F^2$ and $M< F+G$.
Then, for every $M\in [2,F+G-1]$, the quantum query complexity of detecting the existence of a claw for a given function pair $(f,g)$ in $[M]^{[F]}\times [M]^{[G]}$ is lower-bounded by
\[
\Omega\left(\sqrt{G}+F^{1/3}G^{1/6}M^{1/6}\right)
=
\left\{
\begin{array}{lcl}
  \Omega(\sqrt{G})  & \text{if}& 2\le M< (G/F)^2,\\
  \Omega(F^{1/3}G^{1/6}M^{1/6})  & \text{if}& (G/F)^2\le M< F+G.
\end{array}
\right.
\]
\end{theorem}
Note that the condition of $F\le G\le F^2$ and
$M<F+G$ imposed in \autoref{th:ComplexityClawWithSmallerRange_intro}
handles the parameter range that \autoref{th:QueryCompClawWithSmallRange_intro} does not cover.
Setting $F=G=N$ in the theorem recovers the lower bound $\Omega\left(N^{1/2}M^{1/6}\right)$~\cite{AmbBalIra21TQC}.

To evaluate our bound in \autoref{th:ComplexityClawWithSmallerRange_intro},
we will compare it with the lower bound obtained by 
straightforwardly reducing claw detection for two functions with domains of the \emph{same} size and applying the lower bound given in 
~\cite{AmbBalIra21TQC}.
Let us write $\claw_{(a,b)\to c}$ to mean the decision problem of detecting the existence of a claw for given two functions 
in $[c]^{[a]}\times [c]^{[b]}$.
The reduction from $\claw_{(G,G)\to M}$ to $\claw_{(F,G)\to M}$
is as follows, where we assume for simplicity that $G$ is divisible by $F$, i.e., $F|G$ (it is easy to extend the reduction to the general case by padding dummy elements to one of the domains).
Let $(f,g)$ be a pair of input functions to $\claw_{(G,G)\to M}$.
We first partition the domain $[G]$ of $f$ into $G/F$ blocks of size $F$: $B_1,\dots, B_{G/F}$.
For randomly chosen $i\in [G/F]$, we run
a $Q$-query quantum algorithm for $\claw_{(F,G)\to M}$
on $f$ restricted to the domain $B_i$ and $g$,
and apply the quantum amplitude amplification~\cite{BraHoyMosTap02AMS,HoyMosWol03ICALP} to the whole algorithm.
The resulting quantum algorithm makes $O\left(\sqrt{G/F}\cdot Q\right)$ quantum queries.
Since the quantum query complexity of $\claw_{(G,G)\to M}$ for $M< G$ is $\Omega(G^{1/2}M^{1/6})$~\cite{AmbBalIra21TQC}, we have $Q=\Omega(F^{1/2}M^{1/6})$.
Together with trivial lower bound $\Omega(\sqrt{G})$,
we obtain the lower bound $\Omega(\sqrt{G}+F^{1/2}M^{1/6})$. This bound is weaker than our $\Omega(\sqrt{G}+F^{1/3}G^{1/6}M^{1/6})$. 
To see this, assume $G=F^{1+\epsilon}$ for $0<\epsilon \le 1$, satisfying $F<G\le F^2$. 
In this case, we have
\[
F^{\frac{1}{3}}G^{\frac{1}{6}}M^{\frac{1}{6}}=F^{\frac{1}{3}}F^{\frac{1}{6}(1+\epsilon)}M^{\frac{1}{6}}=F^{\frac{1}{2}+\frac{1}{6}\epsilon}M^{\frac{1}{6}}\in \omega\left(F^{\frac{1}{2}}M^{\frac{1}{6}}\right).
\]
For instance, if $F=G^{3/4}$ and $M=G^{3/4}$, then our bound
is $\Omega(\sqrt{G}+F^{1/3}G^{1/6}M^{1/6})=\Omega(G^{13/24})$
while the straightforwardly obtained bound is
$\Omega(\sqrt{G}+F^{1/2}M^{1/6})=\Omega(\sqrt{G})$.

\autoref{tab:claw} summarizes the quantum query complexity for claw detection
for functions $f\colon [F]\to [M]$ and $g\colon [G]\to [M]$, assuming $F\le G$ without loss of generality.
The upper bound $O(G^{1/2+\epsilon}M^{1/4})$ in the case of $M<F+G$ for $F\le G\le F^2$ follows 
from the upper bound~\cite{AmbBalIra21TQC} 
for $\claw_{(G,G)\to M}$
(i.e., claw detection for the same domain sizes) via a simple embedding the domain $[F]$ into $[G]$.

The most interesting open problem is finding the tight bound in the case of $M\le F+G$ for $F\le G\le F^2$.
More precisely, it follows from \autoref{tab:claw}  that the tight bound $\Theta((FG)^{1/3})$ for $M\ge F+G$ also holds for all $M=\Omega(G)$. On the other hand, if $M=\Theta(1)$, then $\Omega(\sqrt{G})$ is tight, since Grover's search algorithm achieves $O(\sqrt{G})$ queries. Thus, the tight bound is open when $M$ is asymptotically smaller than $G$ and larger than $2$.

\subsection{Technical Outline}
\paragraph{Approximate Degree of Multisymmetric Properties.}
The starting point of the proofs of Theorems~\ref{th:bisymmetricproperty_intro} and \ref{th:multisymmetricproperty_intro} is the proof idea for the single-symmetric properties introduced by Ambainis~\cite{Amb05ToC} (we hereafter refer to "single-symmetric" just as "symmetric").
Let $\Phi$ and $\Phi'$ be the symmetric Boolean functions defined in \autoref{subsec:polynomialmethod}.
Recall that $\Phi$ and $\Phi'$ are functions in $\set{x_{ij}\colon (i,j)\in [n]\times [m]}$ and
$\set{x_{ij}\colon (i,j)\in [n]\times [n]}$, respectively.
In addition, 
define $z_j$ as the number of preimages of $j$ via $f$, i.e., 
$z_j=\left|f^{-1}\left(j\right)\right|$.
Then, Ambainis's proof consists of two claims.
The first claim is that the following are equivalent:
\begin{enumerate}
    \item There exists a polynomial $Q$ of degree at most $d$ in $z_1,\ldots,z_m$ that $\epsilon$-approximates $\Phi$.
    \item There exists a polynomial $P$ of degree at most $d$ in $x_{11},\ldots,x_{nm}$ that $\epsilon$-approximates $\Phi$.
\end{enumerate}
The second claim is that the following are equivalent:
\begin{enumerate}
    \item [i] There exists a polynomial $Q\left(z_1,\ldots,z_m\right)$ of degree at most $d$ that $\epsilon$-approximates $\Phi$.
    \item [ii] There exists a polynomial $Q^\prime\left(z_1,\ldots,z_n\right)$ of degree at most  $d$ that $\epsilon$-approximates $\Phi'$. 
\end{enumerate}
These claims together imply that 
the minimum degree of a polynomial that approximates $\Phi$ is at most $d$, if and only if the minimum degree of a polynomial approximating $\Phi^\prime$ is at most $d$.

We want to generalize these two claims to the bisymmetric case (the $k$-symmetric case is a straightforward extension of the bisymmetric case).
Recall in the bisymmetric case that 
$f\colon [F]\to [M]$ is identified with $FM$ Boolean variables $\{x_{ij}:i\in\left[F\right],j\in\left[M\right]\}$, and
$g\colon [G]\to [M]$ is identified with $GM$ Boolean variables 
$\{y_{kj}:k\in\left[G\right],j\in\left[M\right]\}$. 
In addition, we
define $z_j$ as the number of preimages of $j$ via $f$, i.e., $z_j=\left|f^{-1}\left(j\right)\right|$, and
define $w_j$ as the number of preimages of $j$ via $g$, i.e., $w_j=\left|g^{-1}\left(j\right)\right|$.
Let $\Phi$ and $\Phi'$ be the bisymmetric Boolean function 
in $\set{x_{ij}}$ and $\set{y_{kj}}$ as defined in 
\autoref{sec:outcontribution}.
The first claim in the bisymmetric case is that the following are equivalent:
\begin{enumerate}
    \item There exists a polynomial $Q$ of degree at most $d$ in $z_1,\ldots,z_M,w_1,\dots, w_M$ that $\epsilon$-approximates $\Phi$.
    \item There exists a polynomial $P$ of degree at most $d$ in $x_{11},\ldots,x_{FM},y_{11},\dots,y_{GM}$ that $\epsilon$-approximates $\Phi$.
\end{enumerate}
The second claim in the symmetric case is that the following are equivalent
for every $M'\in [F+G,M]$:
\begin{enumerate}
    \item [i] There exists a polynomial $Q\left(z_1,\ldots,z_M,w_1,\dots, w_M\right)$ of degree at most $d$ that $\epsilon$-approximates $\Phi$.
    \item [ii] There exists a polynomial $Q^\prime\left(z_1,\ldots,z_{M'},w_1,\dots,w_{M'}\right)$ of degree at most  $d$ that $\epsilon$-approximates $\Phi'$. 
\end{enumerate}

For the first claim in the symmetric case, it is easy to show that item $1$ implies item $2$ by using the trivial relation $z_j=\sum_{i}x_{ij}$. This also works in the bisymmetric case.
To show the converse in the symmetric case, we consider the average $\Ex_{\pi}[P]$ of $P$ over 
all permutations $\pi$ over the domains of $f$, since $\Ex_{\pi}[P]$ also approximates $\Phi$
due to the symmetry of $\Phi$ over the domains. 
Since we are interested in the degree of the resulting polynomial, it suffices to examine the average of each product term $x_{i_1j_1}\cdots x_{i_kj_k}$ in $P$, and it turns out that
$\Ex_{\pi}\left[x_{i_1j_1}\cdots x_{i_kj_k}\right]$ is a polynomial in $z_j$'s of degree at most $d$. In the bisymmetric case, we take the average 
$\Ex_{\pi_F\times\pi_G}[P]$ of $P$ over 
all permutations $\pi_F$ over the domains of $f$ and
all permutations $\pi_G$ over the domains of $g$.
Note that each product term in $P$ is of the form of 
$x_{i_1 j_1}\cdots x_{i_mj_m}y_{k_1j_1}\cdots y_{k_lj_l}$.
To obtain the average of each term, we apply permutations $\pi_F$ and $\pi_G$ independently, and thus it holds that
\[
\Ex_{\pi_F\times \pi_G}\left[ x_{i_1 j_1}\cdots x_{i_mj_m}y_{k_1j_1}\cdots y_{k_lj_l}\right]
=\Ex_{\pi_F}\left[ x_{i_1 j_1}\cdots x_{i_mj_m}\right]\cdot
\Ex_{\pi_G}\left[ y_{k_1 j_1}\cdots y_{k_lj_l}\right].
\]
It turns out 
$\Ex_{\pi_F}\left[ x_{i_1 j_1}\cdots x_{i_mj_m}\right]$ and 
$\Ex_{\pi_G}\left[ y_{k_1 j_1}\cdots y_{k_lj_l}\right]$
are polynomials in $z_j's$ and $w_j$'s, respectively, of degree $d$
by a similar analysis to the symmetric case.
Consequently, we can obtain the bisymmetric version of the first claim.

For the second claim, it is more involved to
generalize the symmetric case to the bisymmetric case since the permutation over the range is common to the two functions.
It is easy to see that item (i) implies (ii) in the symmetric case since we can obtain $Q'$
by setting $z_{n+1}= \cdots =z_{m}=0$ in $Q$. 
This also works in the bisymmetric case.
For the converse, let us first review the symmetric case.
Since every $z\deq (z_1,\dots ,z_m)$ corresponding to a function $f:\left[n\right]\rightarrow\left[m\right]$ has at most $n$ nonzero coordinates, 
there exists a permutation $\sigma$ over $[m]$ that moves all non-zero coordinates to the first $n$ coordinates, that is,  
$\sigma z=\left(z_{\sigma^{-1}\left(1\right)},\ldots,z_{\sigma^{-1}\left(m\right)}\right)$ satisfies that 
$z_{\sigma^{-1}\left(j\right)}=0$ for every $\sigma^{-1}\left(j\right)\geq n+1$.
This action of $\sigma$ does not change the value of $\Phi$: $\Phi\left(f\right)=\Phi\left(\sigma f\right)$.  
Since $\sigma f$ can be regarded as a function in $\left[n\right]^{\left[n\right]}$, 
it must hold that $\Phi\left(\sigma f\right)=\Phi^\prime\left(\sigma f\right)$.
By the assumption, $\Phi^\prime\left(\sigma f\right)$ 
is approximated by $Q^\prime\left(\sigma z\mid_n\right)$, 
where $\sigma z\mid_n$
is the first $n$ variables of $\sigma z$, i.e., $z_{\sigma^{-1}\left(1\right)},\ldots,z_{\sigma^{-1}\left(n\right)}$. 
We then define from $Q'$ a symmetric polynomial $Q$ in $m$ variables, 
$z_1,\ldots,z_m$, 
such that, for $z$ corresponding to an arbitrary $f$,
\[
Q\left(z\right)=Q\left(\sigma z\right)=Q\left(z_{\sigma^{-1}\left(1\right)},\ldots,z_{\sigma^{-1}\left(n\right)},0,\ldots,0\right)=Q^\prime\left(z_{\sigma^{-1}\left(1\right)},\ldots,z_{\sigma^{-1}\left(n\right)}\right)=Q^\prime\left(\sigma z|_n\right),
\]
where $\sigma$ 
is a particular permutation defined above depending on $f$. 
This implies that 
$Q\left(z\right)$ approximates 
$\Phi\left(f\right)$, 
since $Q^\prime\left(\sigma z\mid_n\right)$ approximates 
$\Phi^\prime\left(\sigma f\right)=\Phi\left(\sigma f\right)=\Phi\left(f\right)$. 
Since we can assume without loss of generality that $Q'$ is a symmetric polynomial,
the fundamental theorem of symmetric polynomials implies that $Q'$ is represented as
a polynomial in the elementary polynomials over $z_1,\dots, z_n$. We then replace
these elementary polynomials in $Q'$ with the corresponding polynomials over
$z_1,\dots, z_m$ to obtain $Q$. 

In the bisymmetric case, $Q$ is a polynomial in two sets of variables:
$z_1,\dots,z_M$ and $w_1,\dots, w_M$;
$Q'$ is a polynomial in two sets of variables
$z_1,\dots,z_{M'}$ and $w_1,\dots, w_{M'}$ such that $F+G\le M'<M$,
since 
there exists a permutation $\sigma$ over $[M]$ that moves all non-zero coordinates 
in $z_1,\dots,z_M$ and $w_1,\dots, w_M$ 
to the first $M'$ coordinates, respectively.
To obtain $Q$ from $Q'$, we utilize a certain symmetry of $Q'$.
However, $Q'$ is \emph{not} invariant under the action of permutations over $[2M']$
to $2M'$ variables $z_1,\dots,z_{M'},w_1,\dots, w_{M'}$,
but it is invariant under the action of $\tau\times \tau$
for permutation $\tau$ over $[M']$ such that 
$(z_1,\dots,z_M)\mapsto
\left(z_{\tau^{-1}\left(1\right)},\ldots,z_{\tau^{-1}\left(M'\right)}\right)$ 
and 
$(w_1,\dots, w_M)
\mapsto
\left(w_{\tau^{-1}\left(1\right)},\ldots,w_{\tau^{-1}\left(M'\right)}\right)$.
Here, the notion of \emph{multisymmetric} polynomials comes in.
Roughly speaking, multisymmetric polynomials are those invariant under the action of
permutations
to the set of vectors of variables, instead of variables. In our case, let us consider
the set of vectors $(z_1,w_1),\dots, (z_{M'},w_{M'})$, each of which consists of two variables.
Then, $Q'$ is invariant under the action of the permutation $\tau$ on this set of vectors.
Now that we have entered the world of multisymmetric polynomials, we can use the fundamental theorem of those polynomials. More concretely, $Q'$
is represented as a polynomial in power-sum polynomials (defined in \autoref{subsec:Multisymmetric Polynomials}) in  $(z_1,w_1),\dots, (z_{M'},w_{M'})$.
Then, we can obtain $Q$ by replacing each occurrence of power-sum polynomials in $Q'$ with 
the corresponding polynomials in  $(z_1,w_1),\dots, (z_{M},w_{M})$.
This approach is easily generalized to the $k$-symmetric case. We just consider 
the symmetry under the action of permutations to the set of vectors $\left(z_1^{(1)},\dots, z_1^{(k)}\right),\dots,\left(z_{M'}^{(1)},\dots, z_{M'}^{(k)}\right)$, 
where $z_i^{(\ell)}=\Abs{f_\ell(j)}$.

\paragraph{Application to Claw Problem}
To provide a lower bound for the claw problem with a small range
using \autoref{th:bisymmetricproperty_intro},
we need a lower bound on the minimum degree
of a polynomial that approximates the claw property 
for some common range size of input functions.
For this, we use Zhang's reduction~\cite{Zha05COCOON}
from the collision problem to the claw problem.
Suppose that there exists a polynomial $P'$ of degree $d$ that 0.1-approximates the claw property for functions 
 $f\in [FG]^{[F]}$ and $g\in [FG]^{[G]}$ such that $F\le G$. Based on $P'$,
we will then construct a polynomial $P$ of degree $d$ that (41/86)-approximates the collision property for functions 
$h\in [M]^{[M]}$
for $M=FG$
with the promise that $h$
 is either one-to-one or two-to-one.
Here we mean by the collision property the Boolean function $\Phi$ which is $\true$ if and only if the input function $h$ is two-to-one.
Since the minimum degree of the latter polynomial is $\Omega(M^{1/3})$~\cite{AarShi04JACM,Amb05ToC}, we obtain $d=\Omega((FG)^{1/3})$. 
\autoref{th:bisymmetricproperty_intro} with this bound implies \autoref{th:DegreeOfClawForSmallRange_intro},
from which we obtain \autoref{th:QueryCompClawWithSmallRange_intro} with the polynomial method.

To obtain polynomial $P$ from $P'$,
let $P'_{ST}$ 
for disjoint subsets $S,T$ of domain $[M]$ of $h$
be a polynomial of degree $d$
(in variables $x_{ij}$ for $i\in S\sqcup T$ and $j\in [M]$)
that $\epsilon$-approximates the claw property on function pairs $[M]^S\times [M]^T$ to $\set{0,1}$.
Then, define a polynomial $P$ as the average of $P'_{ST}$ over uniformly random disjoint subsets $S,T$ of size $F$ and $G$, respectively:
\begin{equation*}
    P\deq \Ex_{S,T\subset [M]\colon S\cap T=\emptyset}[P'_{ST}].
\end{equation*}
Thus, $P$ is a polynomial of degree at most $d$ in variables $x_{ij}$ for $i,j\in [M]$. 

If $h$ is one-to-one,
then $h|_S$ and $h|_T$ obviously have no collision for every $S$ and $T$,
where $h|_S$ and $h|_T$ are the restriction of $h$ to the the subsets $S$ and $T$, respectively, of the domain. Thus, $P$ is close to $0$.
If $h$ is two-to-one, $h|_S$ is injective with high probability
and, in this case, $T$ intersects the preimage of $h|_S(S)$ with high probability.
This follows from the birthday paradox, since $|S|=F<<\sqrt{M}=\sqrt{FG}$ and $|T|=G>>\sqrt{M}$.
Therefore, $P$ is far from $0$. By scaling and shifting $P$, we have a polynomial that $O(1)$-approximates the collision property.

As an application of \autoref{th:QueryCompClawWithSmallRange_intro}, we obtain \autoref{th:ComplexityClawWithSmallerRange_intro} by generalizing the reduction 
from the claw problem for two functions with the \emph{same} domain size~\cite{AmbBalIra21TQC}
to the reduction from the problem for two functions with the \emph{different} domain size:
$f\colon [F]\to [F+G]$ and $g\colon [G]\to [F+G]$ with $F\le G\le F^2$
(the case of $G>F^2$ is already covered by \autoref{th:QueryCompClawWithSmallRange_intro}).
More precisely, we reduce the composition of 
a certain problem (called $\pSearch$~\cite{BraHoyKalKapLapSal19JC}) and the claw problem 
with the function range stated in \autoref{th:QueryCompClawWithSmallRange_intro}
to the claw problem in question, 
and then use the composition theorem of the general adversary bound~\cite{HoyLeeSpa07STOC}.
That is, for every $k,F,G\in \Natural$ with $k|F$ and $G|(F/k)$,
we show the composition $\claw_{(F,G)\to M+2}$ is reducible from the problem
$\claw_{(k,kG/F)\to M}\circ \pSearch_{F/k\to M}$,
where
$\claw_{(a,b)\to c}$ denotes the decision problem of detecting the existence of a claw for given functions $(\phi,\psi)\in  [c]^{[a]}\times [c]^{[b]}$. 
Since $\pSearch_{N\to M}$ is an variant of the unstructured search 
over $N$ elements, the quantum query complexity of $\pSearch_{F/k\to M}$ is $\Omega (\sqrt{F/k})$~\cite{BraHoyKalKapLapSal19JC}. 
By setting $M=k+kG/F$, 
\autoref{th:QueryCompClawWithSmallRange_intro} provides a lower bound on 
the quantum query complexity of $\claw_{(k,kG/F)\to M}$ as
\begin{equation*}
Q_{1/3}\left(\claw_{(k,kG/F)\to M}\right)=\Omega\left(\sqrt{k\frac{G}{F}}+\left(k^2\frac{G}{F}\right)^{1/3}\right).
\end{equation*}
By applying the composition theorem of the general adversary bound,
we have 
\begin{align*}
Q_{1/3}\left(\claw_{(F,G)\to M+2}\right)
&\ge Q_{1/3}\left(\claw_{(k,kG/F)\to M}\circ \pSearch_{F/k\to M}\right)\\
&=\Omega\left(\left[\sqrt{k\frac{G}{F}}+\left(k^2\frac{G}{F}\right)^{1/3}\right]\sqrt{\frac{F}{k}} \right)=\Omega\left( \sqrt{G}+k^{1/6}F^{1/6}G^{1/3} \right).
\end{align*}
We now interpret this in terms of $M$ by substituting $k=M/(G/F+1)=\Theta(MF/G)$ as
\begin{equation*}
Q_{1/3}\left(\claw_{(F,G)\to M+2}\right)
=\Omega\left(\sqrt{G}+\left( \frac{MF}{G}\right)^{1/6} F^{1/6}G^{1/3}\right)
=\Omega\left(\sqrt{G}+F^{1/3}G^{1/6}M^{1/6}
\right)
\end{equation*}
Notice that this holds for very special combinations of $k,F,G$ such that $k|F$ and $G|(F/k)$
and $M=k+kG/F$. However, we can remove this constraint and obtain \autoref{th:ComplexityClawWithSmallerRange_intro}.
\subsection{Organization}
\autoref{sec:Preliminaries} provides notations and necessary notions. \autoref{sec:DegreeOfMultisymmetricProperties} proves \autoref{th:bisymmetricproperty_intro}
and \autoref{th:multisymmetricproperty_intro}.
\autoref{sec:DegreeOfClawProperty} proves 
\autoref{th:DegreeOfClawForSmallRange_intro} 
and
\autoref{th:QueryCompClawWithSmallRange_intro}.
\autoref{sec:ClawProblemWithMuchSmallerRange} proves \autoref{th:ComplexityClawWithSmallerRange_intro}.

\section{Preliminaries}\label{sec:Preliminaries}
\subsection{Notations}\label{subsec:notations}
Let $\Natural$ be the set of natural numbers, $\Integer_+$ the set of non-negative integers, $\Rational$ the set of rational numbers, $\Complex$ the set of complex numbers, and $[n]$  the set $\{1,\dots, n\}$ for every $n\in \Natural$.
If two sets $S$ and $T$ are disjoint, we may write their union as $S\sqcup T$ instead of $S\cup T$ to emphasize their disjointness.

For $n\in \Natural$, let $S_n$ be the group of permutations over $[n]$.
For a commutative ring $R$ and $n\in \Natural$, $R[x_1,\dots ,x_n]$ denotes the ring of polynomials in $x_1,\dots ,x_n$ over $R$. 
For a polynomial $p\in R[x_1,\dots ,x_n]$, 
we write the degree of $p$ 
(i.e., the maximum over all product terms in $p$
of the number of variables in a single product term)
as $\deg(p)$.
We say a polynomial is \emph{multilinear} if the polynomial is linear in any one of the variables
when regarding the other variables as constants.
For any $n\times m$ matrix $\Omega=(\omega_{ij})$ over $\Integer$, $\abs{\Omega}_1$ denotes the sum of the absolute values of all elements in $\Omega$, i.e., $\abs{\Omega}_1\equiv\sum_{i=1}^n\sum_{j=1}^m \abs{\omega_{ij}}$.

For any ordered sets $\calD$ and $\calR$,
we identify a function $f\colon \calD\to \calR$ as the sequence of 
$(f(d_1),\dots, f(d_{|\calD|}))$, where $d_k$ is the $k$th element in $\calD$ for each $k\in [|\calD|]$.
We thus represent by $\calR^{\calD}$ the set of functions from $\calD$ to $\calR$.
We may abbreviate $\calR^{[D]}$ as $\calR^{D}$, when $\calD=[D]$.
For instance, 
$[M]^{[F]}$ (or $[M]^F$) for $F,M\in\Natural$ denotes the set of functions from $\left[F\right]$ to $\left[M\right]$.

For two functions $(f,g)\in [M]^{[F]}\times [M]^{[G]}$, where $F,G,M\in \Natural$, we say that a pair $(x,y)\in [F]\times [G]$ is a \emph{claw} if $f(x)=g(y)$.
These definitions are naturally generalized to the case of more than two functions:
Let $M\in \Natural$ and $F_\ell\in \Natural$ for each $\ell\in [k]$, where $k\in \Natural$ is fixed. 
Then, for functions $f_\ell\in [M]^{[F_\ell]}$ for $\ell\in [k]$, we say that $k$-tuple
$(x_1,\dots, x_k)\in [F_1]\times \dots \times [F_k]$ is a $k$-\emph{claw} if $f_i(x_i)=f_j(x_j)$ for all $i,j\in [k]$.

For a function $f$, 
$\domain(f)$ denotes the domain of $f$,
and $f(S)$ denotes the image of a subset $S\subseteq \domain(f)$ via $f$. 
We can thus write the image of $f$ as $f(\domain(f))$. 

Let $\calD\subseteq \set{0,1}^n$.
For a partial Boolean function $\Phi\colon \calD\to \set{0,1}$ and a non-negative real $\epsilon\in [0,1/2)$, we say that a polynomial $p\colon \Real^n\to \Real$ 
$\epsilon$-approximates $\Phi$ if $p$ satisfies 
$|\Phi(x)-p(x)|\le \epsilon$ and $p(x)\in [0,1]$
for every $x\in \calD$. 
We say the $\epsilon$-approximate degree of $\Phi$, denotedy by
$\adeg_{\epsilon}(\Phi)$,
to mean the minimum degree of a polynomial that 
$\epsilon$-approximates $\Phi$.

The quantum query complexity of a quantum algorithm $\calA$ for $\Phi$ with error $\epsilon$, denoted by $Q^{\calA}_{\epsilon}(\Phi)$ is the number 
of quantum queries
required by $\calA$ to compute $\Phi(x)$ with error probability at most $\epsilon$, when the worst-case input $x\in \calD$ is given as oracle.
The $\epsilon$-error quantum query complexity of $\Phi$, 
$Q_{\epsilon}(\Phi)$,
is the minimum of $Q^{\calA}_{\epsilon}(\Phi)$
over all quantum algorithms $\calA$.
We assume that readers have basic knowledge of quantum computation.

\subsection{Polynomial Method for Quantum Lower Bound}

The polynomial method was introduced by Beals et al.~\cite{BeaBuhCleMosWol01JACM} 
to obtain lower bounds on the number of quantum queries 
required to compute a Boolean function 
on an input bit-string given as an oracle.
It is not difficult to generalize the method to the case where the input is a string of alphabets~\cite{Aar02STOC}.
In the following, we consider the case where the input
consists of two alphabet strings corresponding to two functions.
However, it is not difficult to see similar facts hold for three or more functions.

The following lemma says that, after $q$ quantum queries, any quantum algorithm is in a state where the amplitude of each computational basis state is represented by a polynomial of degree at most $q$ in $x_{ij}$ and $y_{kj}$.
\begin{lemma}\label{lm:oracles}
Let $F,G,M$ be in $\Natural$.
For given oracles $O_f$ and $O_g$ for functions 
$f\in [M]^{[F]}$ and $g\in [M]^{[G]}$,
assume, without loss of generality, that any quantum algorithm with $q$ queries applies 
\[ U_qO_{f/g}U_{q-1}\cdots U_1O_{f/g}U_0\]
to a fixed initial state, say, the all-zero state,
where $U_q,\dots,U_0$ are unitary operators independent of the oracles, and $O_{f/g}$ acts as
\[
\ket{0,i}\ket{b}\ket{z}\rightarrow\ket{0,i}\ket{b+f(i)}\ket{z},\mbox{ and } 
\ket{1,i}\ket{b}\ket{z}\rightarrow\ket{1,i}\ket{b+g(i)}\ket{z},
\]
where $+$ is the addition modulo $M$, or alternatively,
\[
\ket{0,i}\ket{b}\ket{z}\rightarrow\ket{0,i}\ket{b\oplus f(i)}\ket{z},\mbox{ and } 
\ket{1,i}\ket{b}\ket{z}\rightarrow\ket{1,i}\ket{b\oplus g(i)}\ket{z},
\]
where $\oplus$ is bitwise XOR in binary expression. 
At any step, the state of the algorithm is of the form
\[ \sum_{s,i,b,z}{\alpha_{s,i,b,z}\left|s,i\right\rangle\left|b\right\rangle\left|z\right\rangle},\]
where the first and second registers 
consist of $1+\max\set{\ceil{\log F},\ceil{\log G}}$ qubits and $\ceil{\log{M}}$ qubits, respectively
(the last register is a working register that $U_i$ acts on but $O_{f/g}$ does not). 
Define $FM$ variables $x_{ij}\in\{0,1\}$ such that $x_{ij}=1$ if and only if $f\left(i\right)=j$, and $GM$ variables $y_{kj}\in\{0,1\}$ such that $y_{kj}=1$ if and only if $g(k)=j$.  
Then, in the final state, each $\alpha_{s,i,b,z}$ can be represented as a polynomial over $FM+GM$ variables $x_{ij},y_{kj}$ of degree at most $q$.
\end{lemma}
The consequence of the polynomial method is as follows.
\begin{lemma}[~\cite{BeaBuhCleMosWol01JACM,Aar02STOC}]
\label{lm:polynomialmethod_main}
Let $F,G,M$ be in $\Natural$, and
let $\Phi$
be a partial Boolean function over the set of pairs of functions in $\calD\subseteq [M]^{[F]} \times [M]^{[G]}$. 
    For given oracles $O_f$ and $O_g$ for functions 
$(f,g)\in \calD$,
suppose that there exists a quantum algorithm that computes $\Phi(f,g)$ with error $\epsilon$ using $q$ queries. 
Then, there is a polynomial $P$ of degree at most $2q$ in 
variables $x_{11},\dots, x_{FM},y_{11},\dots,y_{GM}$ such that $P$ $\epsilon$-approximates $\Phi$,
where $x_{ij}$ and $y_{kj}$ are the variables defined for $f$ and $g$, respectively, in \autoref{lm:oracles}.
\end{lemma}

The following lemma says that if $\delta$ and $\epsilon$ are constants, then $\adeg_\epsilon(\Phi)$ and $\adeg_\delta (\Phi)$ are linearly related.
\begin{lemma}[\mbox{\cite{BunTha21SIGACT}}]\label{lm:deg_linearly_related}
    For a partial Boolean function $\Phi$ and $0<\epsilon<\delta<1/2 $, it holds that
    \begin{equation*}
    \adeg_\epsilon(\Phi)\le \adeg_\delta (\Phi)\cdot O\left( \frac{\log(1/\epsilon)}{(1/2-\delta)^2} \right).
    \end{equation*}
    In particular, if $\delta$ and $\epsilon$ are constants, then $\adeg_\epsilon(\Phi)$ and $\adeg_\delta (\Phi)$ are linearly related.
\end{lemma}
For completeness, the proofs of the above lemmas are given in \autoref{appdx:polynomialmethod}.

\subsection{Multisymmetric Properties}
For $F,G,M\in \Natural$, let $\Phi$
be a partial Boolean function over the set of pairs of functions in $\calD\subseteq [M]^{[F]}\times [M]^{[G]}$. 
Intuitively, we can think of $\Phi$ as an indicator function deciding whether 
a pair of functions, 
$(f,g)\in \calD$, 
have a certain property.
For instance, we say that $\Phi$ is the \emph{claw} property, 
if $\Phi\left(f,g\right)=1$ if and only if there exists a claw for $(f,g)\in [M]^{[F]}\times [M]^{[G]}$.
We say that $\Phi$ is bisymmetric, if $\Phi$ is invariant 
and $\calD$ is closed
under the permutations over $\domain(f)=[F]$ and $\domain(g)=[G]$, respectively, and the permutation over the common range $[M]$ of $f$ and $g$, i.e., $\left(\sigma f\pi_F,\sigma g\pi_G \right) \in \calD$ if $(f,g)\in \calD$, and
\begin{equation*}
\Phi\left(f,g\right)=\Phi\left(\sigma f\pi_F,\sigma g\pi_G\right),  
\end{equation*}
for all permutations $\pi_F,\pi_G,\sigma$ over $[F]$, $[G]$, and $[M]$, respectively. Here $(\sigma f\pi_F)(i)=\sigma(f(\pi_F(i)))$
and $(\sigma g\pi_G)(j)=\sigma(g(\pi_G(j)))$.
For instance, if $\Phi$ is the claw property, then $\Phi$ is bisymmetric. 

The definition of bisymmetric properties can be generalized 
to the three or more function case
in a straightforward way.
Let $\Phi\colon \calD \to \set{0,1}$
be a partial Boolean function over the set of $k$-tuples of functions in $\calD\subseteq\calF_1\times\cdots \times \calF_k$, where $\calF_{\ell}\ (\ell\in [k])$ denotes $[M]^{[F_{\ell}]}$ for $F_{\ell},M\in \Integer_+$.
We say that $\Phi$ is $k$-symmetric, if $\Phi(f_1,\dots,f_k)$ is invariant 
and $\calD$ is closed 
under the permutations over $\domain(f_{\ell})=[F_\ell]$ for each $\ell\in [k]$, respectively, and the permutation over the common range $[M]$ of $f_{\ell}$'s, i.e., 
$\left(\sigma f_1\pi_{F_1},\dots,\sigma f_k\pi_{F_k}\right)\in \calD$ if $(f_1,\dots,f_k)\in \calD$, and
\begin{equation*}
\Phi(f_1,\dots,f_k)=\Phi\left(\sigma f_1\pi_{F_1},\dots,\sigma f_k\pi_{F_k}\right),  
\end{equation*}
for every permutation $\pi_{F_\ell}$ 
over $[F_\ell]$ for $\ell\in [k]$ and every permutation $\sigma$ over $[M]$, respectively. 

\subsection{Multisymmetric Polynomials}\label{subsec:Multisymmetric Polynomials}
This subsection extends the properties of symmetric polynomials to the case of multisymmetric polynomials.

Recall that a symmetric polynomial over a commutative ring $R$
is a polynomial in $R[x_1,\dots,x_n]$ invariant 
under the action of any permutation $\sigma\in S_n$ on the variables, 
i.e.,  the action of ${x}_i\mapsto {x}_{\sigma(i)}$.
To define multisymmetric polynomials as a generalization of symmetric polynomials, 
we consider a vector consisting of $m$ variables, $X_i\deq (x_{i1},\dots,x_{im})$, 
instead of variable $x_i$,
for each $i\in [n]$. 
Then, a multisymmetric polynomial is defined as a polynomial $p\in R[x_{11},\dots,x_{nm}]$ invariant under the action of any permutation $\sigma\in S_n$ on vectors $X_i$, i.e., the action of $x_{ij}\mapsto x_{\sigma(i)j}$.
Let $R[x_{11},\dots,x_{nm}]^{S_n}$ be the ring of multisymmetric polynomials in 
$R[x_{11},\dots,x_{nm}]$. 
For convenience, let us think of $nm$ variables $x_{11},\dots,x_{nm}$ as 
the $n\times m$ matrix 
$X\deq (X_1^T,\dots,X_n^T)^T$. 
Then, a mutisymmetric polynomial in $R[x_{11},\dots,x_{nm}]^{S_n}$ is invariant
under the action of the permutation on rows of $X$.

For ${\Omega}\deq ({\Omega}_1^T,\dots,{\Omega}_n^T)^T\in \Integer_+^{nm}$ consiting of ${\Omega}_i\deq (\omega_{i1},\dots,\omega_{im}) \in \Integer_+^m$ for each $i\in [n]$, define
\begin{equation*}
    X_i^{{\Omega}_i}\deq \prod_{j\in [m]}x_{ij}^{\omega_{ij}},\ \ \ 
    {X}^{{\Omega}}\deq \prod_{i\in [n]}X_i^{{\Omega}_i}.
\end{equation*}
By the definition, ${X}^{{\Omega}}$ is a monomial of degree $\abs{\Omega}_1=\sum_{i=1}^n\sum_{j=1}^m \omega_{ij}$ (note $\abs{\omega_{ij}}=\omega_{ij}$ for $\omega_{ij}\in \Integer_+$).
We also define  
\begin{equation*}
    \mathrm{mon}_{\Omega}(X)\deq \sum_{\Lambda \in S_n\Omega} {X}^{\Lambda},
\end{equation*}
where $S_n\Omega$ is the orbit of $\Omega$ under the action of $\Omega_i\mapsto \Omega_{\sigma(i)}$ over all permutations $\sigma\in S_n$. 
For instance, if $X^\Omega=x_{12}x_{21}$, then $\mathrm{mon}_{\Omega}(X)=\sum_{i,j\in [n]\colon i\neq j}x_{i2}x_{j1}$.

Observe that the ring $R[x_{11},\dots,x_{nm}]^{S_n}$ is a module over $R$ that has a basis
consisting of $\mathrm{mon}_{\Omega}(X)$ for ${\Omega}\in \Integer_+^{nm}$.
More concretely, any mutisymmetric polynomial of total degree $d$ in $R[x_{11},\dots,x_{nm}]^{S_n}$
is a linear combination of $\mathrm{mon}_{\Omega}(X)$ for $\Omega$'s with $\abs{\Omega}_1\le d$.

Let us define the power-sum $P_\lambda(X)$ over $X$ for  row vector $\lambda\in \Integer_+^m$ as 
\begin{equation*}
    P_\lambda(X)=X_1^{\lambda}+\dots+X_n^{\lambda}.
\end{equation*}
The fundamental theorem of multisymmetric polynomials is as follows.
The proof is given in \autoref{appdx:fundamentaltheoremofmultsymmetric}.
\begin{theorem}[e.g.~\cite{GelKapZel94Book}]\label{th:multsym-polynomial}
    Let $R$ be a commutative ring that contains $\Rational$. For each $\Omega\in \Integer_+^{nm}$,
    $\mathrm{mon}_{\Omega}(X)$ can be expressed as a linear combination over $\Rational$ of products
    $P_{\Omega'_1}(X)\cdots P_{\Omega'_n}(X)$
    of the power-sums over $X=(x_{ij})_{i\in [n],j\in [m]}$ for $\Omega'\deq ((\Omega'_1)^T, \cdots ,(\Omega'_n)^T)^T\in \Integer_+^{nm}$ 
    such that $\abs{\Omega'}_1\le \abs{\Omega}_1$. Moreover,
    the power-sums $P_\lambda(X)$ for $\lambda \in \Integer_+^m$ generate the ring $R[x_{11},\dots,x_{nm}]^{S_n}$ of multisymmetric polynomials.
\end{theorem}

\section{Polynomials Approximating Multisymmetric Properties}
\label{sec:DegreeOfMultisymmetricProperties}
This section proves \autoref{th:bisymmetricproperty_intro}
by generalizing Ambainis's proof idea with \autoref{th:multsym-polynomial}.
\autoref{th:multisymmetricproperty_intro} is a straightfoward extension of \autoref{th:bisymmetricproperty_intro}.
\subsection{Bisymmetric Properties}
We first show that there are two polynomials 
in different sets of variables
of the same degree such that both polynomials approximate $\Phi$.

\begin{lemma}\label{lm:bisymmetryZW}
    For $F,G,M\in \Natural$ and $\calD\subseteq [M]^{[F]}\times [M]^{[G]}$, 
    let $\Phi:\calD\to \set{0,1}$
    be a bisymmetric partial Boolean function over the set of functions 
    $(f,g)\in \calD$.
Define $FM$ variables $x_{ij}$ and $GM$ variables $y_{kj}$ for functions $f$ and $g$
as in \autoref{lm:oracles}. In addition, define $z_j$ as the number of preimages of $j$ via $f$, i.e., $z_j=\left|f^{-1}\left(j\right)\right|$, and
define $w_j$ as the number of preimages of $j$ via $g$, i.e., $w_j=\left|g^{-1}\left(j\right)\right|$.
    Then, the following are equivalent.
    \begin{enumerate}
        \item There exists a polynomial $Q$ of degree at most $d$ in $z_1,\dots ,z_M,w_1,\dots,w_M$ that $\epsilon$-approximates $\Phi$.
        \item There exists a polynomial $P$ of degree at most $d$ in $x_{11},\dots, x_{FM},y_{11},\dots,y_{GM}$ that $\epsilon$-approximates $\Phi$.
    \end{enumerate}
    Moreover, the minimum degree of a polynomial in $z_1,\dots ,z_M,w_1,\dots,w_M$ that $\epsilon$-approximates $\Phi$ is equal to the minimum degree of a polynomial in $x_{11},\dots, x_{FM},y_{11},\dots,y_{GM}$ that $\epsilon$-approximates $\Phi$.
\end{lemma}
\begin{proof}
The first item implies the second, since $z_j=\sum_{i}x_{ij}$ and $w_j=\sum_ky_{kj}$ by the definitions,
and replacing $z_j$ with $\sum_{i}x_{ij}$ and $w_j$ with $\sum_ky_{kj}$ does not change the value and degree of the polynomial.

To show that the second item implies the first, we take the average of $P$, denoted by $\Ex_{\pi_F\times \pi_G}[P]$, over the 
permutations $\pi_F$ of first index $i$ of $x_{ij}$ and the permutations $\pi_G$ of the first index $k$ of $y_{kj}$. 
Since $\Phi$ is invariant under the permutations over the domains of $f,g$, the average $\Ex_{\pi_F\times \pi_G}[P]$ still $\epsilon$-approximates $\Phi$. 
We will show that the expectation is a polynomial over $z_j,w_j$ of degree at most $d$, assuming that the second item holds.

Since $P$ is a linear combination of terms 
$x_{i_1 j_1}\dots x_{i_mj_m}y_{k_1j_1}\dots y_{k_lj_l}$
for positive integers $m$ and $l$ with $m+l\le d$,
it suffices to show that the expectation of any single term is a polynomial over $z_j,w_j$ of degree at most $d$.
Then, the expectation is 
\begin{equation*}
\Ex_{\pi_F\times \pi_G}\left[ x_{i_1 j_1}\cdots x_{i_mj_m}y_{k_1j_1}\cdots y_{k_lj_l}\right]
=\Ex_{\pi_F}\left[ x_{i_1 j_1}\cdots x_{i_mj_m}\right]\cdot
\Ex_{\pi_G}\left[ y_{k_1 j_1}\cdots y_{k_lj_l}\right].
\end{equation*}

Since we define $\set{x_{ij}}$ as the representations of the function $f$, we do not need to consider 
\emph{invalid} assignments to $x_{ij}$'s for which no function can be defined, such as  $x_{ij}=x_{ij'}=1$ for distinct $j,j'$. In addition, since $x_{ij}$ are Boolean, we can assume $ x_{i_1 j_1}\cdots x_{i_mj_m}$ is multilinear.
Therefore, $i_1,\dots,i_m$ are all distinct.
Similarly, $k_1,\dots,k_\ell$ are all distinct.

From the analysis in the single function case \cite{Amb05ToC} assuming \emph{valid} $x_{ij}$ and $y_{kj}$,
it holds that 
$\Ex_{\pi_F}\left[ x_{i_1 j_1}\cdots x_{i_mj_m}\right]$
is a polynomial in $z_j$ of degree at most $m$ and
$\Ex_{\pi_G}\left[ y_{i_1 j_1}\cdots y_{i_lj_l}\right]$
is a polynomial in $w_{j}$ of degree at most $l$ (see Appendix~\ref{appdx:symmetrization} for the completeness). 
Therefore, 
$\Ex_{\pi_F\times \pi_G}\left[ x_{i_1 j_1}\cdots x_{i_mj_m}y_{k_1j_1}\cdots y_{k_lj_l}\right]$
is a polynomial in $z_{j},w_{j}$ of degree $m+l\le d$.
By defining $\Ex_{\pi_F\times \pi_G}[P]$ as $Q$, the second iterm implies the first. This completes the former part of the proof. 

For the latter part of the statement, let $d'$ be the minimum degree of $P$. 
By the former part, the degree of $Q$ is at most $d'$. If the degree of $Q$ were at most $d'-1$, then there would exist a polynomial $P$ of degree at most $d'-1$ that $\epsilon$-approximates $\Phi$, which is a contradiction. Thus, the minimum degree of $Q$ is also $d'$.
\end{proof}
Next, we show that the degree of a polynomial $Q$ shown in \autoref{lm:bisymmetryZW}
does not change, even if we restrict the common range of input functions of $\Phi$ to a small subset.

\begin{lemma}\label{lm:bisymmetryZWandZ'W'}
    For $F,G,M\in \Natural$ such that $M\ge F+G$
    and $\calD\subseteq  [M]^{[F]}\times [M]^{[G]}$,
    let $\Phi\colon \calD\to \set{0,1}$ 
    be a bisymmetric partial Boolean function over the set of function pairs $(f,g)\in\calD$.
    In addition, for every $M'\in \Natural$ such that $F+G\le M'<M$,
    let $\Phi'\colon \calD\cap \left([M']^{[F]}\times [M']^{[G]}\right)\to \set{0,1}$ 
    be the restriction of $\Phi$ to the domain $[M']^{[F]}\times [M']^{[G]}$, i.e., 
    the bisymmetric Boolean function over the set of function pairs
    $(f',g')\in \calD\cap \left([M']^{[F]}\times [M']^{[G]}\right)$
    such that $\Phi'(f',g')=\Phi(f,g)$ for all functions $f',g',f,g$
    such that $f'(x)=f(x)$ and $g'(y)=g(y)$ for every $x\in [F]$ and $y\in [G]$.
    Define $z_j=\left|f^{-1}\left(j\right)\right|$ and $w_j=\left|g^{-1}\left(j\right)\right|$ for $f$ and $g$ and define $z'_j=\left|(f')^{-1}\left(j\right)\right|$ and $w'_j=\left|(g')^{-1}\left(j\right)\right|$ for $f'$ and $g'$. Then, the following are equivalent:
    \begin{enumerate}
        \item There exists a polynomial $Q$ of degree at most $d$ in $z_1,\dots ,z_M,w_1,\dots,w_M$ that $\epsilon$-approximates $\Phi$.
        \item There exists a polynomial $Q'$ of degree at most $d$ in $z'_1,\dots ,z'_{M'},w'_1,\dots,w'_{M'}$ that $\epsilon$-approximates $\Phi'$.
    \end{enumerate}
    Moreover, the minimum degree of a polynomial in $z_1,\dots ,z_M,w_1,\dots,w_M$ that $\epsilon$-approximates $\Phi$ is equal to the minimum degree of a polynomial in $z_{1},\dots, z'_{M'},w'_{1},\dots,w'_{M'}$ that $\epsilon$-approximates $\Phi'$.
\end{lemma}
\begin{proof}
    To show that the first item implies the second, let us think of $f',g'$ as special cases of 
    $(f,g)\in \calD$
    such that both $f([F])$ and $g([G])$ are contained in $[M']$ (note that such $(f',g')$ exists in $\calD$ since $\calD$ is closed under the permutations over the range).
    In this case, we have $z_j=z'_j$ and $w_j=w'_j$ for all $j\in [M']$, and $z_i=w_j=0$ for all $j>M'$. Thus, it holds that 
    \begin{equation*}
    \Phi'(z'_1,\dots,z'_{M'},w'_1,\dots,w'_{M'})=\Phi(z'_1,\dots,z'_{M'},0,\dots ,0,w'_1,\dots,w'_{M'},0,\dots,0).
    \end{equation*}
    Suppose that the first item holds.
    If we define 
    \[
    Q'(z'_1,\dots ,z'_{M'},w'_1,\dots,w'_{M'})\deq Q(z'_1,\dots,z'_{M'},0,\dots ,0,w'_1,\dots,w'_{M'},0,\dots,0),
    \]
    then $Q'$ is a polynomial degree at most $d$ that $\epsilon$-approximates $\Phi'$.
    The second item thus holds.

    We next show that the second item implies the first.
    Fix $(f,g)\in\calD$
    arbitrarily.
    Observe then that there exists a permutation $\sigma$ over $[M]$ such that the images of $\sigma f$
    and $\sigma g$ are contained in $[M']$, since $\domain(f)$ and $\domain(g)$ have cardinality $F$ and $G$, respectively.  
    Let us define $f',g'$ by identifying $(f',g')$ with $(\sigma f,\sigma g)$ as an element in $[M']^{[F]}\times [M']^{[G]}$.
    Then, it holds that $z'_j=z_{\sigma^{-1}(j)}$ and $w'_j=w_{\sigma^{-1}(j)}$ for every $j\in [M']$.
    By bisymmetry, we have $\Phi(f,g)=\Phi(\sigma f, \sigma g)$, which is equal to $\Phi'(f', g')$ by the definition. 
    
    Suppose that the second item holds. Then, $\Phi'(f', g')$ is $\epsilon$-approximated by some polynomial $Q'$ of degree at most $d$ over $z'_1,\dots ,z'_{M'},w'_1,\dots,w'_{M'}$. Since $\Phi'$ is a bisymmetric property, the average of $Q'$ over all permutations over $[M']$, i.e.,
    \begin{equation*}
        \widetilde{Q'}\deq \Ex_{\tau\in S_{M'}}
        \Cset{Q'(z'_{\tau(1)},\dots ,z'_{\tau(M')},w'_{\tau(1)},\dots,w'_{\tau(M')})}
    \end{equation*}
    still $\epsilon$-approximates $\Phi'$. Since $\widetilde{Q'}$ is a bisymmetric polynomial, \autoref{th:multsym-polynomial} implies that $\widetilde{Q'}$ is represented as a polynomial,
    in power-sum polynomials $P'_\lambda$, of degree at most $d$ with respect to variables
    $z'_1,\dots ,z'_{M'},w'_1,\dots,w'_{M'}$,
    where
    \begin{equation*}
    P'_\lambda\deq P'_{\lambda}(z'_1,\dots,z'_{M'},w'_1,\dots, w'_{M'})=\sum_{i=1}^{M'}(z'_i)^{\lambda_z}(w'_i)^{\lambda_w},
    \end{equation*}
    for $\lambda=(\lambda_z,\lambda_w)\in \Integer_+\times \Integer_+$ such that $\lambda_z+\lambda_w\le d$.
    Define $P_\lambda$ by extending $P'_\lambda$ to the case $2M$ variables as
    \begin{equation*}
    P_\lambda\deq P_{\lambda}(z_1,\dots,z_{M},w_1,\dots, w_{M})=\sum_{i=1}^{M}z_i^{\lambda_z}w_i^{\lambda_w}.
    \end{equation*}
    In addition, define a bisymmetric polynomial $Q^{\prime\prime}$ of degree at most $d$ by replacing every occurrence of $P'_\lambda$ in $\widetilde{Q'}$ with $P_\lambda$.
    Then, we have
    \begin{align*}
        Q''(z_1,\dots,z_{M},w_1,\dots,w_{M})&=Q''(z_{\sigma^{-1}(1)},\dots,z_{\sigma^{-1}(M)},w_{\sigma^{-1}(1)},\dots,w_{\sigma^{-1}(M)})\\
        &=Q''(z'_1,\dots,z'_{M'},0,\dots,0,w'_1,\dots,w'_{M'},0,\dots,0)\\
        &=\widetilde{Q'}(z'_1,\dots,z'_{M'},w'_1,\dots,w'_{M'}),
    \end{align*}
    where the first equality follows from the bisymmetry of $Q''$, the second follows from the definition of permutation $\sigma$, and the last follows from
    \begin{equation*}
        P_{\lambda}(z'_1,\dots,z'_{M'},0,\dots,0,w'_1,\dots,w'_{M'},0,\dots,0)=P'_{\lambda}(z'_1,\dots,z'_{M'},w'_1,\dots, w'_{M'}).
    \end{equation*}
    This implies that
    $Q''(z_1,\dots,z_{M},w_1,\dots,w_{M})$ $\epsilon$-approximates $\Phi(z_1,\dots,z_{M},w_1,\dots, w_{M})$, since $\widetilde{Q'}$ $\epsilon$-approximates $\Phi'(z'_1,\dots,z'_{M'},w'_1,\dots, w'_{M'})$, which is equal to $\Phi(z_1,\dots,z_{M},w_1,\dots, w_{M})$ by the definition.
    By setting $Q\deq Q''$, the first item in the statement holds.

    For the latter part of the statement, let $d'$ be the minimum degree of $Q$. 
    By the former part, the degree of $Q'$ is at most $d'$. If the degree of $Q'$ were at most $d'-1$, then there would exist a polynomial $Q$ of degree at most $d'-1$ that $\epsilon$-approximates $\Phi$, which is a contradiction. Thus, the minimum degree of $Q'$ is also $d'$.
\end{proof}

By combining \autoref{lm:bisymmetryZW} with \autoref{lm:bisymmetryZWandZ'W'}, we obtain 
a slightly general version of \autoref{th:bisymmetricproperty_intro}.
\begin{theorem}[general version of \autoref{th:bisymmetricproperty_intro}]
\label{th:bisymmetricproperty}
    For $F,G,M\in \Natural$ such that $M\ge F+G$
    and $\calD\subseteq [M]^{[F]}\times [M]^{[G]}$,
    let 
       $\Phi\colon \calD\to \set{0,1}$ 
    be a bisymmetric Boolean function over the set of function pairs
    $(f,g)\in \calD$.
    In addition, for every $M'\in \Natural$ such that $F+G\le M'<M$, let $\Phi'\colon \calD\cap 
    \left([M']^{[F]}\times [M']^{[G]}\right)\to \set{0,1}$ be
    the restriction of $\Phi$ to the domain $[M']^{[F]}\times [M']^{[G]}$.
        Then, the minimum degree of a polynomial that $\epsilon$-approximates $\Phi$ is equal to the minimum degree of a polynomial that $\epsilon$-approximates $\Phi'$,
    where
    \begin{itemize}
        \item 
        $\Phi$ depends on variables $x_{ij}\in\{0,1\}$ for
        $(i,j)\in [F]\times [M]$ such that $x_{ij}=1$ if and only if $f\left(i\right)=j$, 
        and variables $y_{kj}\in\{0,1\}$ for $(k,j)\in [G]\times [M]$ such that $y_{kj}=1$ if and only if $g(k)=j$, and
        \item 
        $\Phi'$ depends on variables $x'_{ij}\in\{0,1\}$ 
        for $(i,j)\in [F]\times [M']$ such that $x'_{ij}=1$ if and only if $f'\left(i\right)=j$, and variables $y'_{kj}\in\{0,1\}$ for
        $(k,j)\in [G]\times [M']$ such that $y'_{kj}=1$ if and only if $g'(k)=j$.
    \end{itemize}
\end{theorem}

With the promise that the union of the images of input functions $f$ and $g$, i.e., ${f([F])\cup g([G])}$, is upper-bounded in size 
by $\mathfrak{U}$ for every $(f,g)\in\calD$, 
\autoref{lm:bisymmetryZWandZ'W'} holds for every 
$M'$ with $\mathfrak{U}\le M' <M$
with a slight modification to the proof.
Thus, with that promise, \autoref{th:bisymmetricproperty} also holds for such $M'$.
For instance, if
$\calD$ is the set of pairs of two-to-one functions,
then $\abs{f([F])\cup g([G])}\le (F+G)/2$, and thus
\autoref{th:bisymmetricproperty} holds for every 
$M'$ with $(F+G)/2\le M' <M$.
Since $\abs{f([F])\cup g([G])}\le F+G$ always holds, we obtain \autoref{th:bisymmetricproperty} for all $\calD$.

\subsection{\boldmath{$k$}-Symmetric Properties}
This section states a straightforward generalization of
Lemmas~\ref{lm:bisymmetryZW} and \ref{lm:bisymmetryZWandZ'W'}
to the $k$-symmetric properties.
We omit their proofs, since they are almost the same as those in the bisymmetric case.
\autoref{th:multisymmetricproperty} (a slightly general version of \autoref{th:multisymmetricproperty_intro}) follows immediately from the following lemmas.

\begin{lemma}\label{lm:GeneralBisymmetryZW}
    Let $k\ge 2$.
    For $F_{\ell}\in \Natural\ (\ell\in [k])$, $M\in \Natural$,
     let
    $\Phi:\calD\to \set{0,1}$
    be a $k$-symmetric partial Boolean function over the set of $k$-tuples of functions 
    $(f_1,\dots,f_k)\in \calD$,
    where $\calD\subseteq \calF_1\times\dots \times \calF_k$
    and $\calF_{\ell}$ denotes $[M]^{[F_\ell]}$ for each $\ell\in [k]$.
    For each $\ell\in [k]$, $i\in [F_{\ell}]$, and $j\in [M]$, define variable $x^{(\ell)}_{ij}$ such that
$x^{(\ell)}_{ij}=1$ if and only if $f_{\ell}\left(i\right)=j$,
and define $z^{(\ell)}_j$ as the number of preimages of $j\in [M]$ via $f_{\ell}$, i.e., $z^{(\ell)}_j=\left|f_{\ell}^{-1}\left(j\right)\right|$.
    Then, the following are equivalent.
    \begin{enumerate}
        \item There exists a polynomial $Q$ of degree at most $d$ in $kM$ variables $z^{(\ell)}_1,\dots ,z^{(\ell)}_M\ (\ell\in[k])$ that $\epsilon$-approximates $\Phi$.
        \item There exists a polynomial $P$ of degree at most $d$ in $\Oset{\sum_{\ell=1}^kF_{\ell} M}$ variables $x^{(\ell)}_{11},\dots, x^{(\ell)}_{F_\ell M}\ (\ell\in[k])$ that $\epsilon$-approximates $\Phi$.
    \end{enumerate}
Moreover, the minimum degree of a polynomial in $z^{(\ell)}_1,\dots ,z^{(\ell)}_M\ (\ell\in[k])$ that $\epsilon$-approximates $\Phi$ is equal to the minimum degree of a polynomial in $x^{(\ell)}_{11},\dots, x^{(\ell)}_{F_\ell M}\ (\ell\in[k])$ that $\epsilon$-approximates $\Phi$.
\end{lemma}
\begin{lemma}\label{lm:GeneralbisymmetryZWandZ'W'}.
    Let $k\ge 2$.
    For $F_\ell\ (\ell\in [k])$ and $M\in \Natural$ such that $M\ge F_1+\dots +F_k$, 
    let
    $\Phi:\calD\to \set{0,1}$
    be a $k$-symmetric partial Boolean function over the set of $k$-tuples of functions 
    $(f_1,\dots,f_k)\in \calD$, 
    where 
    $\calD\subseteq \calF_1\times\dots \times \calF_k$ and
    $\calF_{\ell}$ denotes $[M]^{[F_\ell]}$ for each $\ell\in [k]$.
    In addition, 
    for every ${M'}\in \Natural$ such that $F_1+\dots+F_k\le {M'}< M$,
    let 
    $\Phi':\calD\cap 
    \left(\calF'_1\times\dots \times \calF'_k\right)\to \set{0,1}$
    be 
    the restriction of $\Phi$ to the domain $[{M'}]^{[F_1]}\times\dots \times [{M'}]^{[F_k]}$, i.e., 
    the $k$-symmetric Boolean function over the set of $k$-tuples of functions 
    $(f'_1,\dots, f'_k)\in \calD\cap 
    \left(\calF'_1\times\dots \times \calF'_k\right)$,
    where $\calF'_\ell$ denotes $[{M'}]^{[F_\ell]}$ for each $\ell\in [k]$,
    such that $\Phi'(f'_1,\dots,f'_k)=\Phi(f_1,\dots,f_k)$ for all functions $f'_{1},\dots,f'_{k},f_1,\dots, f_{k}$
    satisfying $f'_{\ell}(x)=f_{\ell}(x)$ for all $x\in [F_\ell]$ and all $\ell\in [k]$.
    Finally, for each $\ell\in [k]$, define $z^{(\ell)}_j=\left|f_{\ell}^{-1}\left(j\right)\right|$ for each $j\in [M]$
    and ${z'}^{(\ell)}_j=\left|(f'_{\ell})^{-1}\left(j\right)\right|$ for each $j\in [{M'}]$. Then, the following are equivalent:
    \begin{enumerate}
        \item There exists a polynomial $Q$ of degree at most $d$ in $kM$ variables $z^{(\ell)}_1,\dots ,z^{(\ell)}_M\ (\ell\in[k])$ that $\epsilon$-approximates $\Phi$.
        \item There exists a polynomial $Q'$ of degree at most $d$ in ${z'}^{(\ell)}_1,\dots ,{z'}^{(\ell)}_{{M'}}\ (\ell\in[k])$ that $\epsilon$-approximates $\Phi'$.
    \end{enumerate}
\end{lemma}

\begin{theorem}[general version of \autoref{th:multisymmetricproperty_intro}]
\label{th:multisymmetricproperty}
Define $\Phi$ and $\Phi'$ as in~\autoref{lm:GeneralbisymmetryZWandZ'W'},
where
    \begin{itemize}
        \item 
        $\Phi$ depends on $x^{(\ell)}_{ij}\in\{0,1\}$ for 
        $(\ell,i,j)\in [k]\times [F_{\ell}]\times [M]$ such that $x^{(\ell)}_{ij}=1$ if and only if $f_{\ell}\left(i\right)=j$, 
        \item 
        $\Phi'$ depends on ${x'}^{(\ell)}_{ij}\in\{0,1\}$ 
        for $(\ell, i,j)\in [k]\times [F_{\ell}]\times [{M'}]$  such that ${x'}^{(\ell)}_{ij}=1$ if and only if ${f}'_{\ell}\left(i\right)=j$.
      \end{itemize}
      Then, the minimum degree of a polynomial that $\epsilon$-approximates $\Phi$ equals the minimum degree of a polynomial that $\epsilon$-approximates $\Phi'$.
\end{theorem}

\section{Degree of Claw Property}
\label{sec:DegreeOfClawProperty}
This section 
provides a lower bound for the claw problem with a small range
using \autoref{th:bisymmetricproperty_intro}.
For this, we need a lower bound on the minimum degree
of a polynomial that approximates the claw property 
for some common range size of input functions.
To prove the following theorem, 
we use Zhang's reduction~\cite{Zha05COCOON}
from the collision problem 
for an input function that is either one-to-one or two-to-one.

\begin{theorem}\label{th:DegreeOfClawForLargeRange}
For any constant $0<\epsilon<1/2$, the minimum degree of a polynomial that $\epsilon$-approximates the claw property for function pairs in $[FG]^F\times [FG]^G$ is $\Omega\left((FG)^{1/3}\right)$.
\end{theorem}
\begin{proof}
The basic idea is as follows. Suppose that there exists a polynomial $P'$ of degree $d$ that 0.1-approximates the claw property for functions 
 $f\in [M]^{[F]}$ and $g\in [M]^{[G]}$ such that $F\le G$. Based on $P'$,
we will then construct a polynomial $P$ of degree $d$ that (41/86)-approximates the collision property for functions 
$h\in [M]^{[M]}$
for $M=FG$
with the promise that $h$
 is either one-to-one or two-to-one.
Here we mean by the collision property the Boolean function $\Phi$ which is $\true$ if and only if the input function $h$ is two-to-one.
Since the minimum degree of the latter polynomial is $\Omega(M^{1/3})$~\cite{AarShi04JACM,Amb05ToC}, we obtain $d=\Omega(M^{1/3})$. Below, we make this idea more precise.

Take a pair of disjoint subsets $S\subset [M]$ and $T\subset [M]$ of size $F$ and $G$, respectively, and consider the restrictions
of $h$ to $S$ and $T$, 
$h|_S\in [M]^S$ and $h|_T\in [M]^T$,
respectively.
Here, we assume that $h$ depends on Boolean variables $x_{ij}$ for $i,j\in [M]$
such that $x_{ij}=1$ if and only if $h(i)=j$. Accordingly,
$h|_S$ depends on $x_{ij}$ for $i\in S$ and $j\in [M]$; 
$h|_T$ depends on $x_{ij}$ for $i\in T$ and $j\in [M]$.
Let  $p_{\rm cl}(h)$ be the probability 
that $h|_S$ and $h|_T$ have a claw
over a pair of disjoint subsets $S$ and $T$ (of fixed size $F$ and $G$, respectively) chosen uniformly at random.

For disjoint subsets $S,T\subset [M]$,
let $P'_{ST}$ be a polynomial of degree $d$
(in variables $x_{ij}$ for $i\in S\sqcup T$ and $j\in [M]$)
that $\epsilon$-approximates the claw property on function pairs $[M]^S\times [M]^T$.
For every function $h$, let $P'_{ST}(h)$ denote the value of $P'_{ST}$ at the assignment
to $x_{ij}$ corresponding to $h|_S$ and $h|_T$.
Then, define a polynomial $P$ as the average of $P'_{ST}$ over uniformly random disjoint subsets $S,T\subset [M]$ of size $F$ and $G$, respectively:
\begin{equation*}
    P\deq \Ex_{S,T\subset [M]\colon S\cap T=\emptyset}[P'_{ST}].
\end{equation*}
Thus, $P$ is a polynomial of degree at most $d$ in variables $x_{ij}$ for $i,j\in [M]$. Let $P(h)$ denote the average of $P'_{ST}(h)$ over $S,T$ for every function $h$.

If $h$ is one-to-one,
then $h|_S$ and $h|_T$ obviously have no claws for every disjoint $S$ and $T$,
i.e., $p_{\rm cl}(h)=0$.
Thus, it holds $0\le P'_{ST}(h)\le \epsilon$ for all disjoint $S,T\subset [M]$. This implies 
\begin{equation*}
0\le P(h)\le \epsilon.
\end{equation*}
If $h$ is two-to-one, then $P'_{ST}(h)\ge 1-\epsilon$ 
for a $p_{\rm cl}(h)$
fraction of all pairs of disjoint subsets $S$ and $T$. This implies
\begin{equation*}
    P(h)\ge p_{\rm cl}(h) (1-\epsilon).
\end{equation*}
If we set $M=FG$ for $G\ge F\ge 400$ and $\epsilon=0.1$, then 
it follows from \autoref{cl:ub_on_p_cl} that
$P(h)> 0.2\cdot (1-0.1)=0.18$. Therefore, $P$ separates the one-to-one and two-to-one cases with a constant gap. By scaling and shifting $P$, we have a polynomial that $O(1)$-approximates the collision property. For instance, $\tilde{P}\deq (25P+18)/43$ is sufficient and its approximation error is at most $41/86$. Therefore, $\deg(\tilde{P})=\Omega(M^{1/3})=\Omega((FG)^{1/3})$. 
Since $d\ge \deg(P)=\deg(\tilde{P})$,
the minimum degree $d$ of a polynomial that 0.1-approximates the claw property for function pairs $[FG]^F\times [FG]^G$ is $\Omega\left((FG)^{1/3}\right)$.
By \autoref{lm:deg_linearly_related}, the theorem holds.
\end{proof}

\begin{claim}\label{cl:ub_on_p_cl}
    Suppose that $h\in [M]^{[M]}$
    is two-to-one. 
    For $M=FG$ and $G\ge F\ge 400$, it holds that $p_{\rm cl}(h)>0.2$.
\end{claim}
\begin{proof}
    To choose random disjoint subsets $S$ and $T$, we first choose a random subset $S$ of $[M]$ and then choose
    a random subset $T$ of $[M]\setminus S$. 
Observe that, if $h|_S$ is injective and, in addition, $h^{-1}(h(S))\setminus S$ intersects $T$, then a claw exists.
Thus, we obtain
\begin{equation*}
    p_{\rm cl}(h)>
    \Pr_{S,T\subset [M]\colon S\cap T=\emptyset}\Cset{
    (h|_S\text{ is injective})\wedge (T\text{ intersects } h^{-1}(h(S))\setminus S)
    }.
\end{equation*}
We will derive a bound on $M$ that satisfies
\[
\Pr_{S\subset [M]}\Cset{ h|_S\text{ is injective}}\ge 1-\delta\ \text{and}\ 
\Pr_{T\subset [M]\setminus S}\Cset{
     (T\text{ intersects } h^{-1}(h(S))\setminus S)
     \mid
    (h|_S\text{ is injective})
    }\ge 1-\delta', 
\]
    so that
\begin{equation*}
     p_{\rm cl}(h)>(1-\delta)(1-\delta')>1-(\delta+\delta').    
\end{equation*}
    The probability $\Pr_S\Cset{ h|_S\text{ is injective}}$ is 
    \begin{align*}
        \frac{M(M-2)\cdots(M-2(F-1))}{F!\binom{M}{F}}
        &= \frac{M(M-2)\cdots(M-2(F-1))}{M(M-1)\cdots(M-F+1)}\\
        &= 1\cdot \left(1-\frac{1}{M-1}\right)\left(1-\frac{2}{M-2}\right)\cdots \left(1-\frac{F-1}{M-F+1}\right)\\
        &>\prod_{j=1}^{F-1}\exp \left(-\frac{j}{M-2j} \right)\\
        &> \exp\left( -\frac{F(F-1)}{2}\frac{1}{M-2F} \right)
        > \exp\left( -\frac{F^2}{2M-F^2/100} \right),
    \end{align*}
    where the first inequality follows from $1+x\ge \exp(\frac{x}{1+x})$, and
    the last inequality follows from $4F\le F^2/100$ for $F\ge 400$.
    For the last formula to be at least $1-\delta$, we have
\begin{align}\label{eq:lb_on_M}
    M&\ge \frac{F^2}{2}\Cset{
    \frac{1}{\ln \left(\frac{1}{1-\delta}\right)}+\frac{1}{100}
    }.
\end{align}
Next, we consider a bound on $M$ such that
\begin{equation*}
    \Pr_{T\subset [M]\setminus S}\Cset{
     (T\text{ intersects } h^{-1}(h(S))\setminus S)
     \mid
    (h|_S\text{ is injective})
    }
    =1-\frac{\binom{M-2F}{G}}{\binom{M-F}{G}}
\end{equation*}
is lower-bounded by $1-\delta'$. 
Equivalently, we will derive a bound on $M$ such that
${\binom{M-2F}{G}}/{\binom{M-F}{G}}\le \delta'$.
\begin{align*}
    \frac{\binom{M-2F}{G}}{\binom{M-F}{G}}&=\frac{(M-2F)(M-2F-1)\cdots(M-2F-(G-1))}{(M-F)(M-F-1)\cdots(M-F-(G-1))}\\
    &=  \left(1-\frac{F}{M-F}\right)\left(1-\frac{F}{M-F-1}\right)\cdots \left(1-\frac{F}{M-F-(G-1)}\right)\\
    &<\prod_{j=0}^{G-1}\exp\left(-\frac{F}{M-F-j}   \right)
    <\exp \left(-\frac{FG}{M} \right).
\end{align*}    
For this to be at most $\delta'$, we have
\begin{equation}\label{eq:ub_on_M}
    M\le \frac{FG}{\ln({1}/{\delta'})}.
\end{equation}
From Ineqs.~(\ref{eq:lb_on_M}) and (\ref{eq:ub_on_M}), if there exist $M,\delta,\delta'$ such that $\delta+\delta'<1$ and
\begin{equation}\label{eq:lbub_on_M}
\frac{F^2}{2}\Cset{
    \frac{1}{\ln \left(\frac{1}{1-\delta}\right)}+\frac{1}{100}
    }
\le
\frac{FG}{2}\Cset{
    \frac{1}{\ln \left(\frac{1}{1-\delta}\right)}+\frac{1}{100}
    }
\le
M
\le \frac{FG}{\ln({1}/{\delta'})}, 
\end{equation}
then 
$ p_{\rm cl}(h)> 1-(\delta+\delta')>0.$

For $\delta=\delta'=0.4$, 
Ineq.~(\ref{eq:lbub_on_M}) is
$0.98\ldots \times FG\le M\le 1.09...\times FG$.
Therefore, it holds $p_{\rm cl}(h)> 1-(0.4+0.4)=0.2$
for $M=FG$.
\end{proof}
We obtain the following theorem as a corollary of \autoref{th:bisymmetricproperty}.
\begin{theorem}[restatement of \autoref{th:DegreeOfClawForSmallRange_intro}]
\label{th:DegreeOfClawForSmallRange}
Let $F,G\in \Natural$ be such that $F\le G$.
For any constant $0<\epsilon<1/2$ and any $M\ge F+G$, the minimum degree of a polynomial that $\epsilon$-approximates the claw property for function pairs in $[M]^{[F]}\times [M]^{[G]}$ is $\Theta(\sqrt{G}+(FG)^{1/3})$,
where the lower bound $\Omega(\sqrt{G})$ holds for all $M\ge 2$.
\end{theorem}
\begin{proof}
    Let $\Phi$ be the claw property for functions in $[FG]^{[F]}\times [FG]^{[G]}$. By \autoref{th:DegreeOfClawForLargeRange}, we have $\adeg_{\epsilon}(\Phi)=\Omega((FG)^{1/3})$ when $M=FG$.
    By applying \autoref{th:bisymmetricproperty}, we have $\adeg_{\epsilon}(\Phi')=\Omega((FG)^{1/3})$, where $\Phi'$ is the restriction of $\Phi$ to the domain $[M']^{[F]}\times [M']^{[G]}$ for any $M'$ with $F+G\le M' <FG$.
    This proves a $\Omega((FG)^{1/3})$ lower bound in the statement of the theorem. 

    The $\Omega(\sqrt{G})$ lower bound follows using the reduction from the unstructured search~\cite{BuhDurHeiHoyMagSanWol01CCC}: Let $p$ be a polynomial of $\deg_{\epsilon}(\Phi)$ in variables $x_{ij},y_{kj}$ for $(i,k,j)\in [F]\times[G]\times [M]$ that $\epsilon$-approximates the claw property $\Phi$ on a function pair $(f,g)\in [M]^{[F]}\times [M]^{[G]}$. Suppose $M\ge 2$.
    Then consider the special case where
    $f$ maps all inputs to $1$ and $g$ maps all inputs to either $1$ or $2$. This is equivalent to setting $x_{i1}=1$ for all $i\in [F]$, $x_{ij}=0$ for all $i,j$ with $j\ge 2$, and $y_{kj}=0$ for all $k,j$ with $j\ge 3$. 
    Then, the polynomial $p'$ resulting from $p$ with these replacements depends only on $y_{kj}$ over $(k,j)\in [G]\times [2]$. Since we can assume $y_{k2}=1-y_{k1}$ in this special case, $p'$ can be rewritten as a polynomial of the same degree in $y_{k1}$ over $k\in [G]$. 
    Observe now that, if there exists $k$ such that $y_{k1}=1$ (meaning the existence of a claw for $f,g$), 
    it holds $1-\epsilon\le p'\le 1$; otherwise, there exists no claw and thus $0\le p'\le \epsilon$. This implies that $p'$ $\epsilon$-approximates the OR function over $G$ variables $y_{k1}$'s. Because of $\adeg_{\epsilon}(OR_G)=\Theta(\sqrt{G})$~\cite{Pat92STOC}, we have $\adeg_{\epsilon}(\Phi)\ge \deg{p'}=\Omega(\sqrt{G})$.
    
    Since there is an $O\left(\sqrt{G}+(FG)^{1/3}\right)$-query quantum algorithm that detects a claw (and finds it if it exists)~\cite{Tan09TCS}, 
    the polynomial method (\autoref{lm:polynomialmethod_main}) implies that the lower bound $\Omega\left(\sqrt{G}+(FG)^{1/3}\right)$ is tight
    for $M\ge F+G$.
\end{proof}

The proof of \autoref{th:DegreeOfClawForSmallRange} implies the following optimal quantum query complexity.
\begin{theorem}[restatement of \autoref{th:QueryCompClawWithSmallRange_intro}]
\label{th:QueryCompClawWithSmallRange}
Let $F,G\in \Natural$ be such that $F\le G$.
For every $M\ge F+G$,
the (query-)optimal quantum algorithm that computes a claw (or detects the existence of a claw) with a constant error probability
requires $\Theta(\sqrt{G}+(FG)^{1/3})$ queries for a given pair of functions in $[M]^{[F]}\times [M]^{[G]}$.
More precisely, the optimal quantum query complexity is $\Theta ((FG)^{1/3})$ 
for every $M\ge F+G$ if $F\le G\le F^2$ 
and $\Theta(\sqrt{G})$ for every $M\ge 2$ if $G>F^2$.
\end{theorem}

\section{Claw Problem with Much Smaller Range}
\label{sec:ClawProblemWithMuchSmallerRange}
This section proves \autoref{th:ComplexityClawWithSmallerRange_intro}. The proof generalizes the idea used in the case of the same domain size~\cite{AmbBalIra21TQC}: We first reduce the composition of 
a certain problem (called $\pSearch$) and the claw problem for a larger size of the function range, to the claw problem in question, 
and then use the composition theorem of the adversary method.

\begin{theorem}[\cite{HoyLeeSpa07STOC}]\label{th:GeneralAdversaryBound}
    For finite sets $C$ and $E$, and $\calD\subseteq C^n$, let $f\colon \calD\to E$. An adversary matrix
    for $f$ is a $\abs{\calD}\times \abs{\calD}$ real, symmetric matrix $\Gamma$ that satisfies
    $\bra{x}\Gamma\ket{y}=0$ for all $x,y\in \calD$ with $f(x)=f(y)$. Define the general adversary bound as
    \begin{equation*}
        \ADV^{\pm}(f)\deq \max_{\Gamma}\frac{\norm{\Gamma}}{\max_{j\in [n]}\norm{\Gamma\circ \Delta_j}},
    \end{equation*}
    where the maximization is over all possible adversary matrices $\Gamma$ of $f$, and where $\Gamma\circ\Delta_j$ denotes the entry-wise matrix product between $\Gamma$ and $\Delta_j=\sum_{x,y\in D\colon x_j\neq y_j}\ket{x}\bra{y}$. Then, the two-sided $\epsilon$-bounded error quantum query complexity of $f$ is lower bounded by $\ADV^{\pm}(f)$ as follows:
    \begin{equation*}
        Q_\epsilon(f)\ge \frac{1-2\sqrt{\epsilon(1-\epsilon)}}{2}\ADV^{\pm}(f).
    \end{equation*}
\end{theorem}

\begin{definition}[$\pSearch$~\cite{BraHoyKalKapLapSal19JC}] $\pSearch_{K\to M}$ denotes the problem defined as follows.
For $K,M\in \Natural$, given an oracle in $([M]\cup\set\ast)^K$ with the promise that exactly one of the $K$
values is non-$\ast$, the goal is to output the non-$\ast$ value.
\end{definition}
Since $\pSearch_{K\to M}$ can be computed with $O(\sqrt{K})$ queries by using Grover search
and it requires $\Omega(\sqrt{K})$ quantum queries (by the adversary method~\cite{Amb02JCSS}),
\begin{equation}\label{eq:BoundForPsearch}
Q_{1/3}(\pSearch_{K\to M})=\Theta\left(\sqrt{K}\right).
\end{equation}
Brassard et al.~\cite{BraHoyKalKapLapSal19JC} proved the following composition lemma.
\begin{lemma}[\mbox{\cite[Theorem~9]{BraHoyKalKapLapSal19JC}}]
\label{lm:CompAdvPsearch}
    For any function $f\colon [M]^N\to Z$ and
    $\pSearch_{K\to M}$, define 
  $h\colon D\to Z$, where 
  $D\subset ([M]\cup\set\ast)^{KN}$ such that
  $((x_{11},\dots,x_{1K}),\dots, (x_{N1},\dots,x_{NK}))$ is in $D$ if and only if
  exactly one of $x_{i1},\dots,x_{iK}$ is non-$\ast$ for every $i\in [N]$,
  as
$
  h\deq f\circ \pSearch_{K\to M},
$ that is,
  \begin{equation*}
  h\left((x_{11},\dots,x_{1K}),\dots, (x_{N1},\dots,x_{NK})\right)
\deq f\left(\pSearch_{K\to M}(x_{11},\dots,x_{1K}),\dots, \pSearch_{K\to M}(x_{N1},\dots,x_{NK})\right).
\end{equation*}
  Then,
  \[\ADV^{\pm}(h)\ge \frac{2}{\pi}\ADV^{\pm}(f)
  \cdot\ADV^{\pm}(\pSearch_{K\to M}).
  \]
\end{lemma}

By employing \autoref{th:QueryCompClawWithSmallRange} and \autoref{lm:CompAdvPsearch}, we obtain the following theorem.
\begin{theorem}[restatement of \autoref{th:ComplexityClawWithSmallerRange_intro}]
\label{th:ComplexityClawWithSmallerRange}
Let $F,G,M\in \Natural$ be such that $F\le G\le F^2$ and $M< F+G$.%
\footnote{The condition of $F\le G\le F^2$ and $M< F+G$
means the parameter range that \autoref{th:QueryCompClawWithSmallRange} does not cover.}
Then, for every $M\in [2,F+G-1]$, the quantum query complexity of detecting the existence of a claw for a given function pair $(f,g)$ in $[M]^{[F]}\times [M]^{[G]}$ is lower-bounded by
\[
\Omega\left(\sqrt{G}+F^{1/3}G^{1/6}M^{1/6}\right)
=
\left\{
\begin{array}{lcl}
  \Omega(\sqrt{G})  & \text{if}& 2\le M< (G/F)^2,\\
  \Omega(F^{1/3}G^{1/6}M^{1/6})  & \text{if}& (G/F)^2\le M< F+G.
\end{array}
\right.
\]
\end{theorem}

\begin{proof}
We will show that the following problem $\calP$ is reducible to the claw problem
and then give a lower bound on the quantum query complexity of $\calP$.

The input to $\calP$ is a pair of functions 
$f:=(f_1,\dots ,f_F)\in ([M]\cup \set{\ast})^{[F]}$ 
and  
$g:=(g_1,\dots ,g_G)\in ([M]\cup \set{\ast})^{[G]}$ 
(as the sequence representation of functions; see~\autoref{subsec:notations}) for $M\ge 2$
with the following promise, where $f_i$ denotes 
$f(i)$ for $i\in [F]$ and $g_k$ denotes $g(k)$ for $k\in [G]$:
\begin{itemize}
    \item if $f$ is partitioned into $k$ blocks of size $F/k$, 
    \[
    (f_1,\dots,f_{F/k}),\dots,(f_{F-F/k+1},\dots,f_F),
    \]
    then each block contains exactly one non-$\ast$ value;
    \item if $g$ is partitioned into $kG/F$ blocks of size $F/k$, 
    \[
    (g_1,\dots,g_{F/k}),\dots,(g_{G-F/k+1},\dots,g_G),
    \]
    then each block contains exactly one non-$\ast$ value.
\end{itemize}
Here, we assume the promise that $k$ and $F/k$ divide $F$ and $G$, respectively.
This assumption will be removed later.

The goal of $\calP$ is to accept if,
when applying $\pSearch$ on each block of $f$ and $g$,
the collection of the outputs of $\pSearch$ for $f$
intersects the collection of the outputs of $\pSearch$ for $g$,
and reject otherwise. Equivalently, the goal is to detect the existence of a claw
for two functions identified with the following two sequences, respectively:
\[
\left(
    \pSearch_{F/k\to M}(f_1,\dots,f_{F/k}),\dots,\pSearch_{F/k\to M}(f_{F-F/k+1},\dots,f_F)
\right),
\]
\[
\left(
      \pSearch_{F/k\to M}(g_1,\dots,g_{F/k}),\dots,\pSearch_{F/k\to M}(g_{G-F/k+1},\dots,g_G)
\right).
\]
Observe that the problem $\calP$ for $f,g$ with  
$k|F$ and $(k/F)|G$
is reducible to detecting
the existence of a claw for function $f'\in [M+2]^{[F]}$ and  $g'\in [M+2]^{[G]}$,
where $f'$ and $g'$ are obtained by replacing every occurrence of $\ast$ with $M+1$ and $M+2$, respectively. That is, for each $(f,g)$ with the promise and the corresponding $(f',g')$, it holds that 
\[
\claw_{(F,G)\to M+2}(f',g')=\left(\claw_{(k,kG/F)\to M}\circ \pSearch_{F/k\to M}\right)(f,g),
\]
where $\claw_{(a,b)\to c}$is the decision problem of detecting the existence of a claw for given functions $(\phi,\psi)\in  [c]^{[a]}\times [c]^{[b]}$.
Note that, since each query access to $f'\ (g')$ can be simulated by a single query access to $f\ (\text{respectively},\ g)$, the modification to input functions needs no extra queries.
Therefore, the quantum query complexity of $\claw_{(F,G)\to M+2}$ is lower-bounded
by the quantum query complexity of $\claw_{(k,kG/F)\to M}\circ \pSearch_{F/k\to M}$.

Let us set $M=k+kG/F$, for which
we will bound the quantum query complexity of $\claw_{(k,kG/F)\to M}$ by using
\autoref{th:QueryCompClawWithSmallRange}:
If $k^2<kG/F$ (i.e., $k<G/F$), then $Q(\claw_{(k,kG/F)\to M})=\Omega(\sqrt{kG/F})$.
If $kG/F\le k^2$ (i.e., $G/F\le k$), then $Q(\claw_{(k,kG/F)\to M})=\Omega\left((k\cdot kG/F)^{1/3}\right)$.
More succinctly, 
\begin{equation}\label{eq:claw_(k,kG/F)->M}
Q_{1/3}\left(\claw_{(k,kG/F)\to M}\right)=\Omega\left(\sqrt{k\frac{G}{F}}+\left(k^2\frac{G}{F}\right)^{1/3}\right).
\end{equation}
By applying \autoref{th:GeneralAdversaryBound} and \autoref{lm:CompAdvPsearch} 
with \autoref{eq:BoundForPsearch} and \autoref{eq:claw_(k,kG/F)->M},
we have 
\begin{align*}
Q_{1/3}\left(\claw_{(F,G)\to M+2}\right)
&\ge Q_{1/3}\left(\claw_{(k,kG/F)\to M}\circ \pSearch_{F/k\to M}\right)\\
&=\Omega\left(\left[\sqrt{k\frac{G}{F}}+\left(k^2\frac{G}{F}\right)^{1/3}\right]\sqrt{\frac{F}{k}} \right)=\Omega\left( \sqrt{G}+k^{1/6}F^{1/6}G^{1/3} \right).
\end{align*}
We now interpret this in terms of $M$ by substituting $k=M/(G/F+1)=\Theta(MF/G)$ as
\begin{equation}\label{eq:LowerBoundOnClawWithAssuption}
Q_{1/3}\left(\claw_{(F,G)\to M+2}\right)
=\Omega\left(\sqrt{G}+\left( \frac{MF}{G}\right)^{1/6} F^{1/6}G^{1/3}\right)
=\Omega\left(\sqrt{G}+F^{1/3}G^{1/6}M^{1/6}
\right),
\end{equation}
where $M=k+kG/F$.

To remove the assumption that $k$ and $F/k$ divides $F$ and $G$, respectively,
let us consider any $F,G,k$ such that $F\le G\le F^2$ and $1\le k\le F$. 
Define
    \begin{equation*}
    F'\deq k\Floor{\frac{F}{k}},\ \ \ \ \ \ \
    G'\deq \frac{F'}{k}\Floor{\frac{G}{F'/k}}.
    \end{equation*}
Note that
    \begin{equation*}
        G'=\frac{F'}{k}\Floor{k\frac{G}{F'}}\ge \frac{F'}{k}k=F'.
    \end{equation*}
     By the definition, $k$ divides $F'$ and $F'/k$ divides $G'$. 
    Then, for any instance $f'\colon [F']\to [M+2]$ and $g'\colon [G']\to [M+2]$ of $\claw_{(F',G')\to M+2}$, we define $f\colon [F]\to [M+2]$ and $g\colon [G]\to [M+2]$ as 
    \begin{itemize}
        \item $f(x)\deq f'(x)$ for all $x\in [F']$, and $f(x)\deq f'(F')$ for all $x\in [F]\setminus [F']$;
        \item $g(y)\deq g'(y)$ for all $y\in [G']$, and $g(y)\deq g'(G')$ for all $y\in [G]\setminus [G']$.
    \end{itemize}   
 Then, one can see that there exists a claw for $\left(f^\prime,g^\prime\right)$ if and only if there exists a claw for $\left(f,g\right)$. 
 For $M=k+kG^\prime/F^\prime$, we thus have
\begin{align*}
Q_{1/3}\left(\claw_{\left(F,G\right)\rightarrow M+2}\right)&\geq         
Q_{1/3}\left(\claw_{\left(F^\prime,G^\prime\right)\rightarrow M+2}\right)\\
&=\Omega\left(\sqrt 
G+\left(F^\prime\right)^{1/3}({G^\prime)\ }^{1/6}M^{1/6}\right)=\Omega\left(\sqrt 
G+F^{1/3}G^{1/6}M^{1/6}\right),
\end{align*}
where the first equality follows from \autoref{eq:LowerBoundOnClawWithAssuption}
and 
the last equality follows from $F^\prime\geq F/2$ and $G^\prime\geq G/2$ as shown in \autoref{cl:M_k}.

Next, we will remove the assumption that $M$ is of the form
$M=k+kG^\prime/F^\prime$ and show that the lower bound holds for 
every integer $M$ such that $2\le M < F+G$.
For this, let us arbitrarily fix $F$ and $G$ such that $F\le G\le F^2$. 
For every integer $M$ such that 
$M\in\left[1+\left\lfloor G/F\right\rfloor,F+G\right]$, 
there exists $k\in [1,F]$ such that $M_k\le M\le M_{k+1}$ by \autoref{cl:M_k},
where $M_k$ is defined in the claim.
Since $\claw_{\left(F,G\right)\rightarrow M_k+2}$ is reducible to $\claw_{\left(F,G\right)\rightarrow M+2}$ since $[M_k+2]\subseteq [M+2]$, we have 
\begin{equation*}
Q_{1/3}\left(\claw_{\left(F,G\right)\rightarrow M+2}\right)
\geq     
Q_{1/3}\left(\claw_{\left(F,G\right)\rightarrow M_k+2}\right)
=\Omega\left(\sqrt G+F^{1/3}G^{1/6}M_k^{1/6}\right)
=\Omega\left(\sqrt G+F^{1/3}G^{1/6}M^{1/6}\right),
\end{equation*}
where the last equality follows from $M/M_k<M_{k+1}/M_k\le 8$ by \autoref{cl:M_k}. 
By replacing $M+2$ with $M$, we have
\begin{equation*}
Q_{1/3}\left(\claw_{\left(F,G\right)\rightarrow M}\right)
=\Omega\left(\sqrt G+F^{1/3}G^{1/6}M^{1/6}\right)
\end{equation*}
for $3+\left\lfloor G/F\right\rfloor\le M\le F+G$.
Since it holds that $\Omega\left(\sqrt G+F^{1/3}G^{1/6}M^{1/6}\right)=\Omega\left(\sqrt G\right)$ 
for every positive $M<3+\left\lfloor G/F\right\rfloor$, and it follows from the unstructured search that
$Q_{1/3}\left(\claw_{\left(F,G\right)\rightarrow M}\right)=\Omega\left(\sqrt{ G}\right)$~\cite{BuhDurHeiHoyMagSanWol01CCC}
for every $M\geq 2$, it holds that 
$Q_{1/3}\left(\claw_{\left(F,G\right)\rightarrow M}\right)
=\Omega\left(\sqrt G+F^{1/3}G^{1/6}M^{1/6}\right)$ 
for $2\le M\le F+G$.
More precisely, 
Since $\sqrt G>F^{1/3}G^{1/6}M^{1/6}$ if and only if $\left(G/F\right)^2>M$, we have%
\footnote{The case $2\le M <(G/F)^2$ should be ignored when $G/F\le \sqrt{2}$.}
\begin{equation*}
    Q_{1/3}\left(\claw_{\left(F,G\right)\rightarrow M}\right)=
    \left\{
    \begin{array}{lcl}
       \Omega\left(\sqrt G\right)  & \text{if} & 2\le M<\left(G/F\right)^2 \\
        \Omega\left(F^{1/3}G^{1/6}M^{1/6}\right) & \text{if} & \left(G/F\right)^2\le M<F+G.
    \end{array}
    \right.
\end{equation*}
\end{proof}
\begin{claim}\label{cl:M_k}
    For every $F,G\in \Natural$ such that $F\le G\le F^2$, let us define $M_k\deq k+kG_k^\prime/F_k^\prime$ 
    for every $1\le k\le F$, where $F_k^\prime=k\left\lfloor F/k\right\rfloor$ 
    and $G_k^\prime=(F_k^\prime/k\ )\left\lfloor\frac{G}{\left(F_k^\prime/k\right)}\right\rfloor$. 
Then, it holds that 
\begin{enumerate}
    \item $F/2\le F_k^\prime\le F$, and $G/2\le G_k^\prime\le G$ for every $k\in\left[1,F\right]$,
    \item $M_k\le M_{k+1}$ for every $k\in\left[1,F-1\right]$,
    \item $M_1=1+\left\lfloor G/F\right\rfloor$ and $M_F=F+G$,
    \item $M_{k+1}/M_k\le8$ for every $k\in\left[1,F-1\right]$.
\end{enumerate}
\end{claim}
\begin{proof}
We first prove the second and third items.
Since 
$kG_k^\prime/F_k^\prime
=\left\lfloor\frac{G}{\left(F^\prime/k\right)}\right\rfloor$ 
is the number of blocks of size $F_k^\prime/k$ in $G^\prime$ and the block size $F_k^\prime/k=\left\lfloor F/k\right\rfloor$ 
is monotone decreasing with respect to $k$, $kG_k^\prime/F_k^\prime$ is monotone increasing. Therefore, $M_k$ is monotone increasing with respect to $k$. Therefore, 
$M_k\le M_{k+1}$ for every $k\in\left[1,F-1\right]$. 
Since 
$F_1^\prime=F$ and $G_1^\prime=\left\lfloor G/F\right\rfloor$, 
we have $M_1=1+\left\lfloor G/F\right\rfloor$. 
Since $F_F^\prime=F$ and $G_F^\prime=G$, 
we have $M_F=F+G$. 

We next prove the first and the last items. 
Since $F_k^\prime\le F$ and $G_k^\prime\le G$ 
by the definition, the first item follows from the lower bounds
on $F'_k$ and $G'_k$.
By the definition of $F_k^\prime$, we have 
$F-F_k^\prime<\left\lfloor F/k\right\rfloor$ 
and $F_k^\prime\geq\left\lfloor F/k\right\rfloor$. 
Hence,
\begin{equation*}
F_k^\prime
\geq\max{\left\{F-\left\lfloor\frac{F}{k}\right\rfloor,\left\lfloor\frac{F}{k}\right\rfloor\right\}}
\geq\frac{F}{2}.
\end{equation*}
By the definition of $G_k^\prime$, 
we have $G-G_k^\prime<F_k^\prime/k=\left\lfloor F/k\right\rfloor$ 
and $G_k^\prime\geq F_k^\prime/k =\left\lfloor F/k\right\rfloor$. Hence,
\begin{equation*}
G_k^\prime
\geq\max{\left\{G-\left\lfloor\frac{F}{k}\right\rfloor,\left\lfloor\frac{F}{k}\right\rfloor\right\}}
\geq\frac{G}{2}.
\end{equation*}
Together with $F_k^\prime\le F$ and $G_k^\prime\le G$, we have
\begin{equation*}
\frac{1}{2}\cdot\frac{G}{F}
\le\frac{G_k^\prime}{F_k^\prime}\le2\frac{G}{F}.
\end{equation*}
Then, the last item follows from
\begin{equation*}
\frac{M_{k+1}}{M_k}
=\frac{k+1}{k}\frac{1+G_{k+1}^\prime/F_{k+1}^\prime}{1+G_k^\prime/F_k^\prime}
\le 2\cdot\frac{1+2G/F}{1+G/\left(2F\right)}\le8.
\end{equation*}
\end{proof}

\section*{Acknowledgments}
This work was partially supported by JSPS KAKENHI Grant Numbers JP20H05966 and JP22H00522.

\bibliographystyle{plain}
\bibliography{ClawSmallRange}

\begin{thebibliography}{10}

\bibitem{Aar02STOC}
Scott Aaronson.
\newblock Quantum lower bound for the collision problem.
\newblock In {\em Proceedings of the Thirty-Fourth Annual ACM Symposium on
  Theory of Computing}, pages 635--642. ACM, 2002.

\bibitem{AarShi04JACM}
Scott Aaronson and Yaoyun Shi.
\newblock Quantum lower bounds for the collision and the element distinctness
  problems.
\newblock {\em Journal of the ACM}, 51(4):595--605, 2004.

\bibitem{Amb02JCSS}
Andris Ambainis.
\newblock Quantum lower bounds by quantum arguments.
\newblock {\em Journal of Computer and System Sciences}, 64(4):750--767, 2002.

\bibitem{Amb05ToC}
Andris Ambainis.
\newblock Polynomial degree and lower bounds in quantum complexity: Collision
  and element distinctness with small range.
\newblock {\em Theory of Computing}, 1(1):37--46, 2005.

\bibitem{Amb06JCSS}
Andris Ambainis.
\newblock Polynomial degree vs. quantum query complexity.
\newblock {\em Journal of Computer and System Sciences}, 72(2):220--238, 2006.

\bibitem{Amb07SICOMP}
Andris Ambainis.
\newblock Quantum walk algorithm for element distinctness.
\newblock {\em SIAM Journal on Computing}, 37(1):210--239, 2007.

\bibitem{AmbBalIra21TQC}
Andris Ambainis, Kaspars Balodis, and J\={a}nis Iraids.
\newblock {A Note About Claw Function with a Small Range}.
\newblock In {\em 16th Conference on the Theory of Quantum Computation,
  Communication and Cryptography (TQC 2021)}, volume 197 of {\em Leibniz
  International Proceedings in Informatics (LIPIcs)}, pages 6:1--6:5. Schloss
  Dagstuhl -- Leibniz-Zentrum f{\"u}r Informatik, 2021.

\bibitem{AmbGilJefKok20STOC}
Andris Ambainis, Andr{\'{a}}s Gily{\'{e}}n, Stacey Jeffery, and Martins
  Kokainis.
\newblock Quadratic speedup for finding marked vertices by quantum walks.
\newblock In {\em Proceedings of the 52nd Annual {ACM} {SIGACT} Symposium on
  Theory of Computing, {STOC} 2020, Chicago, IL, USA, June 22-26, 2020}, pages
  412--424. {ACM}, 2020.

\bibitem{BeaBuhCleMosWol01JACM}
Robert Beals, Harry Buhrman, Richard Cleve, Michele Mosca, and Ronald de~Wolf.
\newblock Quantum lower bounds by polynomials.
\newblock {\em Journal of the ACM}, 48(4):778--797, 2001.

\bibitem{BelChiJefKotMag13ICALP}
Aleksandrs Belovs, Andrew~M. Childs, Stacey Jeffery, Robin Kothari, and
  Fr{\'{e}}d{\'{e}}ric Magniez.
\newblock Time-efficient quantum walks for $3$-distinctness.
\newblock In {\em Proceedings of the 40th International Colloquium on Automata,
  Languages, and Programming, {ICALP} 2013, Part {I}}, pages 105--122, 2013.
\newblock See \url{http://arxiv.org/abs/1302.3143} and
  \url{http://arxiv.org/abs/1302.7316}.

\bibitem{BraHoyKalKapLapSal19JC}
Gilles Brassard, Peter H{\o}yer, Kassem Kalach, Marc Kaplan, Sophie Laplante,
  and Louis Salvail.
\newblock Key establishment {\`a} la merkle in a quantum world.
\newblock {\em Journal of Cryptology}, 32(3):601--634, 2019.

\bibitem{BraHoyMosTap02AMS}
Gilles Brassard, Peter H{\o}yer, Michele Mosca, and Alain Tapp.
\newblock Quantum amplitude amplification and estimation.
\newblock In {\em Quantum Computation and Information}, volume 305 of {\em
  Contemporary Mathematics}, pages 53--74. American Mathematical Society, 2002.

\bibitem{BraHOyTap98LATIN}
Gilles Brassard, Peter H{\O}yer, and Alain Tapp.
\newblock Quantum cryptanalysis of hash and claw-free functions.
\newblock In {\em Proceedings of Third Latin Americal Symposium on Theoretical
  Informatics (LATIN'98)}, pages 163--169. Springer Berlin Heidelberg, 1998.

\bibitem{BuhDurHeiHoyMagSanWol01CCC}
Harry Buhrman, Christoph D{\"u}rr, Mark Heiligman, Peter H{\o}yer,
  Fr{\'e}d{\'e}ric Magniez, Miklos Santha, and Ronald de~Wolf.
\newblock Quantum algorithms for element distinctness.
\newblock In {\em Proceedings of the 16th IEEE Conference on Computational
  Complexity}, pages 131--137, 2001.

\bibitem{BunKotTha20TOC}
Mark Bun, Robin Kothari, and Justin Thaler.
\newblock The polynomial method strikes back: Tight quantum query bounds via
  dual polynomials.
\newblock {\em Theory of Computing}, 16(10):1--71, 2020.

\bibitem{BunTha21SIGACT}
Mark Bun and Justin Thaler.
\newblock Guest column: Approximate degree in classical and quantum computing.
\newblock {\em SIGACT News}, 51(4):48--72, Jan. 2021.

\bibitem{GelKapZel94Book}
Israel~M. Gelfand, Mikhail~M. Kapranov, and Andrei~V. Zelevinsky.
\newblock {\em Discriminants, Resultants, and Multidimensional Determinants}.
\newblock Mathematics: Theory \& Applications. Birkh{\"a}user Boston, 1994.

\bibitem{HosSasTanXag20TCS}
Akinori Hosoyamada, Yu~Sasaki, Seiichiro Tani, and Keita Xagawa.
\newblock Quantum algorithm for the multicollision problem.
\newblock {\em Theor. Comput. Sci.}, 842:100--117, 2020.

\bibitem{HoyLeeSpa07STOC}
Peter H{\o}yer, Troy Lee, and Robert Spalek.
\newblock Negative weights make adversaries stronger.
\newblock In {\em Proceedings of the Thirty-Ninth Annual ACM Symposium on
  Theory of Computing}, pages 526--535. ACM, 2007.

\bibitem{HoyMosWol03ICALP}
Peter H{\o}yer, Michele Mosca, and Ronald de~Wolf.
\newblock Quantum search on bounded-error inputs.
\newblock In {\em Proceedings of Thirtieth International Colloquium on
  Automata, Languages and Programming (ICALP 2003)}, volume 2719 of {\em
  Lecture Notes in Computer Science}, pages 291--299, 2003.

\bibitem{IwaKaw03NG}
Kazuo Iwama and Akinori Kawachi.
\newblock A new quantum claw-finding algorithm for three functions.
\newblock {\em New Generation Computing}, 21(4):319--327, 2003.

\bibitem{JaqSch19CRYPTO}
Samuel Jaques and John~M. Schanck.
\newblock Quantum cryptanalysis in the ram model: Claw-finding attacks on sike.
\newblock In {\em Proceedings of Advances in Cryptology -- CRYPTO 2019}, pages
  32--61. Springer International Publishing, 2019.

\bibitem{JefZur23STOC}
Stacey Jeffery and Sebastian Zur.
\newblock Multidimensional quantum walks.
\newblock In {\em Proceedings of the 55th Annual {ACM} Symposium on Theory of
  Computing, {STOC} 2023, Orlando, FL, USA, June 20-23, 2023}, pages
  1125--1130. {ACM}, 2023.

\bibitem{JouLuc09ASIACRYPT}
Antoine Joux and Stefan Lucks.
\newblock Improved generic algorithms for $3$-collisions.
\newblock In {\em Proceedings of the 15th International Conference on the
  Theory and Application of Cryptology and Information Security,
  ASIACRYPT~2009}, pages 347--363, 2009.
\newblock See \url{https://eprint.iacr.org/2009/305}.

\bibitem{Kut05ToC}
Samuel Kutin.
\newblock Quantum lower bound for the collision problem with small range.
\newblock {\em Theory of Computing}, 1(1):29--36, 2005.

\bibitem{LiuZha19EUROCRYPTO}
Qipeng Liu and Mark Zhandry.
\newblock On finding quantum multi-collisions.
\newblock In {\em Proceedings of the 38th Annual International Conference on
  the Theory and Applications of Cryptographic Techniques, {EUROCRYPT} 2019,
  Part {III}}, pages 189--218, 2019.
\newblock See \url{https://eprint.iacr.org/2018/1096}.

\bibitem{MagNayRolSan11SICOMP}
Fr{\'e}d{\'e}ric Magniez, Ashwin Nayak, J{\'e}r{\'e}mie Roland, and Miklos
  Santha.
\newblock Search via quantum walk.
\newblock {\em SIAM J. Comput.}, 40(1):142--164, 2011.

\bibitem{NikSas16ASIACRYPT}
Ivica Nikoli\'{c} and Yu~Sasaki.
\newblock A new algorithm for the unbalanced meet-in-the-middle problem.
\newblock In {\em Proceedings of the 22nd International Conference on the
  Theory and Application of Cryptology and Information Security, {ASIACRYPT}
  2016, Part {I}}, pages 627--647, 2016.
\newblock See \url{https://eprint.iacr.org/2016/851}.

\bibitem{Pat92STOC}
Ramamohan Paturi.
\newblock On the degree of polynomials that approximate symmetric boolean
  functions (preliminary version).
\newblock In {\em Proceedings of the Thirty-Fourth Annual ACM Symposium on
  Theory of Computing}, pages 468--474. ACM, 1992.

\bibitem{Sze04FOCS}
Mario Szegedy.
\newblock Quantum speed-up of markov chain based algorithms.
\newblock In {\em Proceedings of the Forty-Fifth Annual IEEE Symposium on
  Foundations of Computer Science}, pages 32--41. IEEE Computer Society, 2004.

\bibitem{Tan09TCS}
Seiichiro Tani.
\newblock Claw finding algorithms using quantum walk.
\newblock {\em Theoretical Computer Science}, 410(50):5285--5297, 2009.

\bibitem{Zha15QIC}
Mark Zhandry.
\newblock A note on the quantum collision and set equality problems.
\newblock {\em Quantum Info. Comput.}, 15(7-8):557--567, 2015.

\bibitem{Zha05COCOON}
Shengyu Zhang.
\newblock Promised and distributed quantum search.
\newblock In {\em Proceedings of International Computing and Combinatorics
  Conference (COCOON 2005)}, pages 430--439. Springer Berlin Heidelberg, 2005.

\end{thebibliography}
\appendix
\section*{Appendix}
\section{Claim of Polynomial Method}\label{appdx:polynomialmethod}
\begin{lemma}[restatement of \autoref{lm:oracles}]
Let $F,G,M$ be in $\Natural$.
For given oracles $O_f$ and $O_g$ for functions 
$f\in [M]^{[F]}$ and $g\in [M]^{[G]}$,
assume without loss of generality that any quantum algorithm with $q$ queries applies 
\[ U_qO_{f/g}U_{q-1}\cdots U_1O_{f/g}U_0\]
to the initial state, say, the all-zero state,
where $U_q,\dots,U_0$ are unitary operators independent of the oracles, and $O_{f/g}$ acts as
\[
\ket{0,i}\ket{b}\ket{z}\rightarrow\ket{0,i}\ket{b+f(i)}\ket{z},\mbox{ and } 
\ket{1,i}\ket{b}\ket{z}\rightarrow\ket{1,i}\ket{b+g(i)}\ket{z},
\]
where $+$ is the addition modulo $M$, or alternatively,
\[
\ket{0,i}\ket{b}\ket{z}\rightarrow\ket{0,i}\ket{b\oplus f(i)}\ket{z},\mbox{ and } 
\ket{1,i}\ket{b}\ket{z}\rightarrow\ket{1,i}\ket{b\oplus g(i)}\ket{z},
\]
where $\oplus$ is bitwise XOR in binary expression. 
At any step, the state of the algorithm is of the form
\[ \sum_{s,i,b,z}{\alpha_{s,i,b,z}\left|s,i\right\rangle\left|b\right\rangle\left|z\right\rangle},\]
where the first and second registers 
consist of $1+\max\set{\ceil{\log F},\ceil{\log G}}$ qubits and $\ceil{\log{M}}$ qubits, respectively
(The last register is a working register that $U$ acts on but $O_{f/g}$ does not). 
Define $FM$ variables $x_{ij}\in\{0,1\}$ such that $x_{ij}=1$ if and only if $f\left(i\right)=j$, and $GM$ variables $y_{kj}\in\{0,1\}$ such that $y_{kj}=1$ if and only if $g(k)=j$.  
Then, in the final state, each $\alpha_{s,i,b,z}$ can be represented as a polynomial over $FM+GM$ variables $x_{ij},y_{kj}$ of degree at most $q$, where $q$ is the number of quantum queries.
\end{lemma}
\begin{proof}
For $q=0$, since $\alpha_{s,i,b,z}$ is independent of the oracles, the statement holds.
Suppose that the claim holds for $q-1$ and let
$\sum_{s,i,b,z}{\alpha_{s,i,b,z}\left|s,i\right\rangle\left|b\right\rangle\left|z\right\rangle}$
be the state after applying $U_{q-1}$, where each $\alpha_{s,i,b,z}$ is a polynomial over $x_{ij},y_{ij}$ of degree at most $q-1$.
Assume that the $q$th $O_{f/g}$ maps this state to
$\sum_{s,i,b,z}{\beta_{s,i,b,z}\left|s,i\right\rangle\left|b\right\rangle\left|z\right\rangle}$.
Observe that $O_{f/g}$ maps each basis state 
$\ket{0,i}\ket{b}\ket{z}$ to $\ket{0,i}\ket{b+f(i)}\ket{z}$ 
(likewise $\ket{1,k}\ket{b}\ket{z}$) to $\ket{0,k}\ket{b+g(k)}\ket{z}$),
and that $\ket{0,i}\ket{b-j+f(i)}\ket{z}$ is exactly $\ket{0,i}\ket{b}\ket{z}$ if $x_{ij}=1$ (and likewise $\ket{1,k}\ket{b-j+g(k)}\ket{z}$ is exactly $\ket{1,k}\ket{b}\ket{z}$ if $y_{kj}=1$).
Thus, it holds that
\begin{align*}
    \beta_{0,i,b,z}&=x_{i1}\alpha_{0,i,b-1,z}+x_{i2}\alpha_{0,i,b-2,z}+\dots +x_{iM} \alpha_{0,i,b-M,z},\\
\beta_{1,k,b,z}&=y_{k1} \alpha_{1,k,b-1,z}+y_{k2} \alpha_{1,k,b-2,z}+\dots +y_{kM}\alpha_{1,k,b-M,z}.
\end{align*}
Therefore, $\beta_{0,i,b,z}$ and $\beta_{1,k,b,z}$ are polynomials over $x_{ij}$ and $y_{ij}$ of degree at most $q$. After applying unitary $U_q$, the amplitude of each basis state is a linear combination of $\beta_{0,i,b,z}$ and $\beta_{1,k,b,z}$. Therefore, its degree is still at most $q$.  
In the case of the XOR-type oracle, the same argument holds.
\end{proof}
\begin{lemma}[restatement of \autoref{lm:polynomialmethod_main}]
Let $F,G,M$ be in $\Natural$, and
let $\Phi$
be a Boolean function over the set of pairs of functions in $[M]^{[F]} \times [M]^{[G]}$. 
    For given oracles $O_f$ and $O_g$ for functions 
    $f\in [M]^{[F]}$ and $g\in [M]^{[G]}$,
suppose that there exists a quantum algorithm that computes $\Phi(f,g)$ with error $\epsilon$ using $q$ queries. 
Then, there is a polynomial $P$ of degree at most $2q$ in 
variables $x_{11},\dots, x_{FM},y_{11},\dots,y_{GM}$ such that $P$ $\epsilon$-approximates $\Phi$,
where $x_{ij}$ and $y_{kj}$ are the variables defined for $f$ and $g$, respectively, in \autoref{lm:oracles}.
\end{lemma}
\begin{proof}
    Assume, without loss of generality, that the quantum algorithm measures the last single-qubit register in the final state 
    $\sum_{s,i,b,z}{\alpha_{s,i,b,z}\left|s,i\right\rangle\left|b\right\rangle\left|z\right\rangle}$
    in the computational basis and outputs the outcome as the value of $\Phi$. 
    Let $A$ be the set of tuples $(s,i,b,z)$ with the last bit of $z$ being $\Phi(f,g)$.
    Then, the probability that the algorithm outputs $\Phi(f,g)$ is 
    \begin{equation*}
    p_{\Phi}(f,g)\deq\sum_{(s,i,b,z)\in A}\Abs{\alpha_{s,i,b,z}}^2.
    \end{equation*}
    It follows from \autoref{lm:oracles} that $p_{\Phi}$ is represented as a polynomial
    of degree at most $2q$ in  $x_{11},\dots, x_{FM},y_{11},\dots,y_{GM}$.
    Since the error probability of the algorithm is at most $\epsilon$, we have 
    $\Abs{p_{\Phi}(f,g)-\Phi(f,g)}\le \epsilon$
    for all $f,g$. This implies that there is a polynomial $P$ of degree at most $2q$ in 
variables $x_{11},\dots, x_{FM},y_{11},\dots,y_{GM}$ such that $P$ $\epsilon$-approximates $\Phi$.
\end{proof}
\begin{lemma}[\mbox{\cite{BunTha21SIGACT}}; restatement of \autoref{lm:deg_linearly_related}
]
    For a partial Boolean function $\Phi$ and $0<\epsilon<\delta<1/2 $, it holds that
    \begin{equation*}
    \adeg_\epsilon(\Phi)\le \adeg_\delta (\Phi)\cdot O\left( \frac{\log(1/\epsilon)}{(1/2-\delta)^2} \right).
    \end{equation*}
    In particular, if $\delta$ and $\epsilon$ are constants, then $\adeg_\epsilon(\Phi)$ and $\adeg_\delta (\Phi)$ are linearly related.
\end{lemma}
\begin{proof}
    Let $p$ be a polynomial of $\adeg_\delta (\Phi)$ that $\delta$-approximate $\Phi$. Define a degree-$(\ell\cdot \deg_\delta (\Phi))$ polynomial $A(p)$ as
    \begin{equation*}
        A(p)\deq \sum_{k\ge \ell/2}\binom{\ell}{k}p^k(1-p)^{\ell-k}.
    \end{equation*}
    Intuitively, this is the probability of seeing at least $\ell/2$ heads among $\ell$ coin flips, if $p$ is the probability of seeing heads when flipping a single coin. By Chernoff's bound, the approximation error of $A(p)$ is $\exp (-2\ell (1/2-\delta)^2)$. For this to be $\epsilon$, we have 
     \begin{equation*}
    \ell= O\left( \frac{\log(1/\epsilon)}{(1/2-\delta)^2} \right).
    \end{equation*}
    The latter part follows from the trivial fact that $\adeg_\delta (\Phi)\le \adeg_\epsilon(\Phi)$.
    \end{proof}
\section{Fundamental Theorem of Multisymmetric Polynomials}\label{appdx:fundamentaltheoremofmultsymmetric}
\begin{theorem}[e.g.~\cite{GelKapZel94Book}; restatement of~\autoref{th:multsym-polynomial}]
    Let $R$ be a commutative ring that contains $\Rational$. For each $\Omega\in \Integer_+^{nm}$,
    $\mathrm{mon}_{\Omega}(X)$ can be expressed as a linear combination over $\Rational$ of products
    of the power-sums, $P_{\Omega'_1}(X)\cdots P_{\Omega'_n}(X)$ over $X=(x_{ij})_{i\in [n],j\in [m]}$, for $\Omega'\deq ((\Omega'_1)^T, \cdots ,(\Omega'_n)^T)^T\in \Integer_+^{nm}$ 
    such that $\abs{\Omega'}_1\le \abs{\Omega}_1$. Moreover,
    the power-sums $P_\lambda(X)$ for $\lambda \in \Integer_+^m$ generate the ring $R[x_{11},\dots,x_{nm}]^{S_n}$ of multisymmetric polynomials.
\end{theorem}
\begin{proof}
    Since every mutisymmetric polynomial of total degree $d$ in $R[x_{11},\dots,x_{nm}]^{S_n}$ is generated by $\mathrm{mon}_{\Omega}(X)$ for $\Omega$'s with $\abs{\Omega}_1\le d$, the former part implies that every multisymmetric polynomial can be expressed as a polynomial in the power-sums $P_\lambda(X)$ over $R$. Conversely, every polynomial generated by power-sums $P_\lambda(X)$ is obviously in $R[x_{11},\dots,x_{nm}]^{S_n}$. Therefore, the latter part of the statement holds. It thus suffices to prove the former part.
    
    For each $\Omega=(\Omega_1^T,\dots,\Omega_n^T)^T$, consider the product of power-sums:
    \begin{equation}\label{eq:ProductOfPowersum}
        P_{\Omega_1}(X)\cdots P_{\Omega_n}(X)=(X_1^{\Omega_1}+\cdots +X_n^{\Omega_1})\cdots (X_1^{\Omega_n}+\cdots +X_n^{\Omega_n}).
    \end{equation}
    If $\Omega=O$, then $P_{\Omega_1}(X)\cdots P_{\Omega_n}(X)=n^n$ and $\mathrm{mon}_{\Omega}(X)=1$. Thus, it holds that $\mathrm{mon}_{\Omega}(X)=(1/n^n)P_{\Omega_1}(X)\cdots P_{\Omega_n}(X)$.
    
    Suppose $\Omega\neq O$. In this case, there is at least one non-zero row.
    We claim that $P_{\Omega_1}(X)\cdots P_{\Omega_n}(X)$ is represented as
    \begin{equation}\label{claim:powersum}
        P_{\Omega_1}(X)\cdots P_{\Omega_n}(X)=c_\Omega\cdot \mathrm{mon}_{\Omega}(X)+\sum_{\Omega'}c_{\Omega'}\cdot \mathrm{mon}_{\Omega'}(X),
    \end{equation}
    where $c_\Omega$ and $c_{\Omega'}$ are positive integers (possibly depending on $n$), and   
    the latter summation is taken over $\Omega'=((\Omega'_1)^T,\dots,(\Omega'_n)^T)^T$ such that (1) the number of non-zero rows in $\Omega'$ is strictly less than the number of those in $\Omega$, 
    (2) $\abs{\Omega}_1=\abs{\Omega'}$.

    With \autoref{claim:powersum}, a simple induction on the number of non-zero rows of $\Omega$ completes the proof: Suppose that $\Omega$ has a single non-zero row, say $\Omega_i$. In this case,
    \begin{equation*}
        P_{\Omega_1}(X)\cdots P_{\Omega_n}(X)=n^{n-1}(X_1^{\Omega_i}+\dots+X_n^{\Omega_i})=n^{n-1} \cdot \mathrm{mon}_{\Omega}(X).
    \end{equation*}
    For $\Omega$ with $k\ge 2$ non-zero rows, it follows from the claim that
\begin{equation*}
    \mathrm{mon}_{\Omega}(X)=\frac{1}{c_\Omega}P_{\Omega_1}(X)\cdots P_{\Omega_n}(X)-\sum_{\Omega'}\frac{c_{\Omega'}}{c_{\Omega}}\cdot \mathrm{mon}_{\Omega'}(X).
\end{equation*}
    By the induction hypothesis, this implies that $\mathrm{mon}_{\Omega}(X)$ is a linear combination of $P_{\Omega'_1}(X)\cdots P_{\Omega'_n}(X)$ 
    with rational coefficients
    over $\Omega'=((\Omega'_1)^T,\dots,(\Omega'_n)^T)^T$ such that $\abs{\Omega'}_1\le \abs{\Omega}_1$. This proves the statement in the lemma.

   We now prove \autoref{claim:powersum}. First, consider the sum of the products obtained by taking a term
   for distinct row vector $X_i$ from each clause in the R.H.S. of \autoref{eq:ProductOfPowersum}.
    The sum is represented as 
    \begin{equation*}
     \sum_{\sigma\in S_n}X_1^{\Omega_{\sigma(1)}}\dots X_n^{\Omega_{\sigma(n)}}.  
    \end{equation*}
    By the orbit-stabilizer theorem, 
    each distinct term in this sum appears exactly $\abs{(S_n)_{\Omega}}$ times, where $(S_n)_{\Omega}$ is the subgroup of $S_n$ that stabilizes $\Omega$.
    Thus, the sum is represented as
    \begin{equation*}
     \Abs{(S_n)_{\Omega}}\sum_{\Lambda\in S_n\Omega}X_1^{{\Lambda_1}}\dots X_n^{\Lambda_{n}}=\Abs{(S_n)_{\Omega}}\cdot \mathrm{mon}_{\Omega}(X).
    \end{equation*}
    Next, consider the remaining products, i.e., the products obtained by taking terms
    for the same row of $X$ from multiple clause in the L.H.S. of \autoref{eq:ProductOfPowersum}. For instance, consider the product obtained by taking $X_i^{\Omega_1}$ and  $X_i^{\Omega_2}$ from the first and second clauses, respectively, and a term for distinct row vector $X_j\ (i\neq j)$ from each of the other clause. Then, the product is of the form of $X_i^{\Omega_1}X_i^{\Omega_2}Q$, where $Q$ is the product of remaining variables. This can be regarded as $X^{\Omega'}$ with $\Omega'_i=\Omega_1+\Omega_2$. In general, one can observe that such a product is represented as $X^{\Omega'}$, where
    each non-zero row in $\Omega'$ is the sum of a disjoint set of non-zero rows in $\Omega$ so that the number of non-zero rows in $\Omega'$ is strictly less than the number of those in $\Omega$, and $\abs{\Omega}_1=\abs{\Omega'}$. It is not difficult to see that the set of $\Omega'$ induced by products obtained by expanding R.H.S. of \autoref{eq:ProductOfPowersum} is partitioned into orbits under the action of $S_n$ such that $\Omega'_i\to \Omega'_{\sigma(i)}$ for each $i\in [n]$ and $\sigma\in S_n$. For any representative $\Omega'$ of each orbit, the sum of the products in \autoref{eq:ProductOfPowersum} associated with the orbit is expressed by the orbit-stabilizer theorem as
     \begin{equation*}
     \Abs{(S_n)_{\Omega'}}\sum_{\Lambda\in S_n\Omega'}X_1^{{\Lambda_1}}\dots X_n^{\Lambda_{n}}=\Abs{(S_n)_{\Omega'}}\cdot \mathrm{mon}_{\Omega'}(X).
    \end{equation*}
\end{proof}
\section{Symmetrization}\label{appdx:symmetrization}
For any function $f\in [M]^{[F]}$, 
the expectation of $x_{\pi(i_1)j_1}\cdots x_{\pi(i_m)j_m}$
over all permutations $\pi$
is 
\begin{align*}
\Ex_{\pi}\left[x_{\pi(i_1)j_1}\cdots x_{\pi(i_m)j_m}\right]&=1\cdot\Pr_{\pi}
{\left[x_{\pi(i_1)j_1}\cdots x_{\pi(i_m)j_m}=1\right]}\\
&=\Pr_\pi{\left[x_{\pi(i_1)j_1}=1\right]}\cdot\prod_{\ell=2}^{m}
\Pr_{\pi}{\left[x_{\pi(i_\ell) j_\ell}
=1\mid x_{\pi(i_1)j_1}\cdots x_{\pi(i_{\ell-1})j_{\ell-1\ }}=1\right]},
\end{align*}
where the expectation and the probability are taken over the permutation $\pi$ over $[F]$.
Then, $\Pr_{\pi}{\left[x_{\pi(i_1)j_1}=1\right]}=z_{j_1}/F$
follows from the definition of $z_{j_1}$. 
Next,
we impose on $\pi$ the condition that $f\left(\pi(i_1)\right)=j_1,\ldots,f\left(\pi(i_{\ell-1})\right)=j_{\ell-1}$, i.e.,
$x_{i_1j_1}\cdots x_{i_{\ell-1}j_{\ell-1\ }}=1$. 
Then, there exist $F-\left(\ell-1\right)$ choices for $\pi(i_\ell)$, since the term does not contain the pair of variables with the same first index. To satisfy $x_{\pi(i_\ell) j_\ell}=1$, it must hold that $\pi(i_\ell)\in f^{-1}\left(j_\ell\right)\setminus
\{\pi(i_1),\ldots,\pi(i_{\ell-1})\}$.
Let $s_\ell$ be 
the number of indices in $\{\pi(i_1),\ldots,\pi(i_{\ell-1})\}$ such
that they are mapped via $f$ to $j_\ell$.
Then, there are $z_{j_\ell}-s_\ell$ choices 
for $\pi(i_\ell)$ satisfying $x_{\pi(i_\ell) j_\ell}=1$.
Therefore, 
\[
\Pr{\left[x_{i_\ell j_\ell}=1\mid x_{i_1j_1}\cdots x_{i_{\ell-1}j_{\ell-1\ }}=1\right]}=\frac{z_{j_\ell}-s_\ell}{F-\left(\ell-1\right)}.
\]
Hence, $\Ex_{\pi}\left[x_{\pi(i_1)j_1}\cdots x_{\pi(i_m)j_m}\right]=\prod_{\ell=1}^{m}\frac{z_{j_\ell}-s_\ell}{F-\left(\ell-1\right)}$,
which is a polynomial in $z_j$'s of degree at most $m$.

\end{document}